\documentclass[11pt]{article}
\usepackage[utf8]{inputenc}

\usepackage{amsmath,amssymb,amsfonts,amsthm}
\usepackage{graphicx}
\usepackage{multirow}
\usepackage{adjustbox}
\usepackage{hyperref}
\usepackage{color}
\usepackage{wrapfig}
\usepackage{tikz}
\usetikzlibrary{decorations.pathreplacing}
\usepackage{setspace}
\usepackage{algorithm}
\usepackage[noend]{algpseudocode}
\usepackage{xspace}
\usepackage{pgfplots}
\usepackage{framed}
\usepackage{subcaption}
\usepackage{thmtools}
\usepackage{thm-restate}
\pgfplotsset{compat=1.5}
\usepackage[margin=2in]{geometry}

\newtheorem{theorem}{Theorem}[section]
\newtheorem{corollary}[theorem]{Corollary}
\newtheorem{lemma}[theorem]{Lemma}

\newtheorem{definition}[theorem]{Definition}
\newtheorem{remark}[theorem]{Remark}
\newtheorem{claim}[theorem]{Claim}

\newtheorem{fact}[theorem]{Fact}

\newenvironment{proofof}[1]{\begin{trivlist} \item {\bf Proof
#1:~~}}
  {\qed\end{trivlist}}

\newcommand{\namedref}[2]{\hyperref[#2]{#1~\ref*{#2}}}
\newcommand{\thmlab}[1]{\label{thm:#1}}
\newcommand{\thmref}[1]{\namedref{Theorem}{thm:#1}}
\newcommand{\lemlab}[1]{\label{lem:#1}}
\newcommand{\lemref}[1]{\namedref{Lemma}{lem:#1}}

\newcommand{\corlab}[1]{\label{cor:#1}}
\newcommand{\corref}[1]{\namedref{Corollary}{cor:#1}}
\newcommand{\seclab}[1]{\label{sec:#1}}
\newcommand{\secref}[1]{\namedref{Section}{sec:#1}}

\newcommand{\appref}[1]{\namedref{Appendix}{app:#1}}
\newcommand{\factlab}[1]{\label{fact:#1}}
\newcommand{\factref}[1]{\namedref{Fact}{fact:#1}}
\newcommand{\remlab}[1]{\label{rem:#1}}
\newcommand{\remref}[1]{\namedref{Remark}{rem:#1}}

\newcommand{\alglab}[1]{\label{alg:#1}}
\renewcommand{\algref}[1]{\namedref{Algorithm}{alg:#1}}
\newcommand{\tablelab}[1]{\label{tab:#1}}
\newcommand{\tableref}[1]{\namedref{Table}{tab:#1}}
\newcommand{\deflab}[1]{\label{def:#1}}
\newcommand{\defref}[1]{\namedref{Definition}{def:#1}}

\def \Lap    {\mdef{\mathsf{Lap}}}
\def \C    {\mdef{\mathsf{C}}}

\def \dist    {\mdef{\mathsf{dist}}}

\newcommand{\PPr}[1]{\ensuremath{\mathbf{Pr}\left[#1\right]}}

\renewcommand{\O}[1]{\ensuremath{\mathcal{O}\left(#1\right)}}
\newcommand{\tO}[1]{\ensuremath{\tilde{\mathcal{O}}\left(#1\right)}}
\newcommand{\eps}{\varepsilon}

\def \Lap    {\mdef{\mathsf{Lap}}}

\def \calA    {\mdef{\mathcal{A}}}

\def \calD    {\mdef{\mathcal{D}}}

\def \calS    {\mdef{\mathcal{S}}}
\def \calT    {\mdef{\mathcal{T}}}
\def \calU    {\mdef{\mathcal{U}}}
\def \calW    {\mdef{\mathcal{W}}}

\def \calR    {\mdef{\mathcal{R}}}
\def \frakS    {\mdef{\mathfrak{S}}}

\newcommand{\mdef}[1]{{\ensuremath{#1}}\xspace}  
\newcommand{\myfunc}[1]{\mdef{\mathsf{#1}}}      

\DeclareMathOperator*{\polylog}{polylog}
\DeclareMathOperator*{\poly}{poly}
\DeclareMathOperator*{\negl}{negl}

\def \negl     {\mdef{\myfunc{negl}}}                

\newcommand{\ignore}[1]{}

\newif\ifnotes\notestrue 
\ifnotes
\newcommand{\samson}[1]{\textcolor{blue}{{\bf (Samson:} {#1}{\bf ) }} \marginpar{\tiny\bf
             \begin{minipage}[t]{0.5in}
               \raggedright S:
            \end{minipage}}}   
\newcommand{\jeremiah}[1]{\textcolor{red}{{\bf (Jeremiah:} {#1}{\bf ) }} \marginpar{\tiny\bf
             \begin{minipage}[t]{0.5in}
               \raggedright J:
            \end{minipage}}}

\newcommand{\tanote}[1]{\textcolor{purple}{{\bf (Tamalika:} {#1}{\bf ) }} \marginpar{\tiny\bf
             \begin{minipage}[t]{0.5in}
               \raggedright T:
            \end{minipage}}}
\newcommand{\enote}[1]{\textcolor{green}{{\bf (Elena:} {#1}{\bf ) }} \marginpar{\tiny\bf
             \begin{minipage}[t]{0.5in}
               \raggedright E:
            \end{minipage}}}
\newcommand{\seunghoon}[1]{\textcolor{red}{{\bf (Seunghoon:} {#1}{\bf ) }} \marginpar{\tiny\bf
             \begin{minipage}[t]{0.5in}
               \raggedright T:
            \end{minipage}}}

\else
\newcommand{\samson}[1]{}
\newcommand{\jeremiah}[1]{}
\newcommand{\tanote}[1]{}
\newcommand{\seunghoon}[1]{}
\newcommand{\enote}[1]{}
\fi

\providecommand{\email}[1]{\href{mailto:#1}{\nolinkurl{#1}\xspace}}

\definecolor{darkpastelgreen}{rgb}{0.01, 0.75, 0.24}
\definecolor{bleudefrance}{rgb}{0.19, 0.55, 0.91}
\hypersetup{
     colorlinks   = true,
     citecolor    = bleudefrance,
	 linkcolor	  = darkpastelgreen,
	 urlcolor     = darkpastelgreen
}

\makeatletter
\renewcommand*{\@fnsymbol}[1]{\textcolor{darkpastelgreen}{\ensuremath{\ifcase#1\or *\or \dagger\or \ddagger\or
 \mathsection\or \triangledown\or \mathparagraph\or \|\or **\or \dagger\dagger
   \or \ddagger\ddagger \else\@ctrerr\fi}}}
\makeatother
\usepackage{fullpage}

\title{How to Make Your Approximation Algorithm Private:\\
\vspace{.1in}
\large	A Black-Box Differentially-Private Transformation\\ for Tunable Approximation Algorithms of Functions with Low Sensitivity}

\author{
Jeremiah Blocki\thanks{E-mail: \email{jblocki@purdue.edu}. Supported in part by NSF CCF-1910659, NSF CNS-1931443. and NSF CAREER award CNS-2047272}
\\Purdue University
\and
Elena Grigorescu\thanks{E-mail: \email{elena-g@purdue.edu}. Supported in part by NSF CCF-1910659, NSF CCF-1910411, and NSF CCF-2228814.}
\\Purdue University
\and
Tamalika Mukherjee\thanks{E-mail: \email{tm3391@columbia.edu}. Most of the work was done as a Ph.D. student at Purdue University. Supported in part by the Bilsland Dissertation Fellowship, NSF CCF-1910659 and and NSF CCF-2228814.}
\\Columbia University
\and
Samson Zhou\thanks{E-mail: \email{samsonzhou@gmail.com}. Work done in part while at Carnegie Mellon University and supported in part by a Simons Investigator Award and by NSF CCF-1815840.}
\\UC Berkeley and Rice University
}

\date{\today}
\begin{document}

\maketitle
\begin{abstract}

We develop a framework for efficiently transforming certain approximation algorithms into differentially-private variants, in a black-box manner. Specifically, our results focus on algorithms $A$ that output an approximation to a function $f$ of the form $(1-\alpha)f(x)-\kappa\leq A(x) \leq (1+\alpha)f(x)+\kappa$, where $\kappa \in \mathbb{R}_{\geq 0}$ denotes additive error and $\alpha\in [0,1)$ denotes multiplicative error can be``tuned" to small-enough values while incurring only a polynomial blowup in the running time/space. We show that such algorithms can be made differentially private without sacrificing accuracy, as long as the function $f$ has small ``global sensitivity''. We achieve these results by applying the ``smooth sensitivity'' framework developed by Nissim, Raskhodnikova, and Smith (STOC 2007). 

Our framework naturally applies to transform non-private FPRAS and FPTAS algorithms into $\varepsilon$-differentially private approximation algorithms where the former case requires an additional postprocessing step. We apply our framework in the context of sublinear-time and sublinear-space algorithms, while preserving the nature of the algorithm in meaningful ranges of the parameters. Our results include the first (to the best of our knowledge) $\eps$-edge differentially-private sublinear-time algorithm for estimating the number of triangles, the number of connected components, and the weight of a minimum spanning tree of a graph whose accuracy holds with high probability. 
In the area of streaming algorithms, our results include $\eps$-DP algorithms for estimating $L_p$-norms, distinct elements, and weighted minimum spanning tree for both insertion-only and turnstile streams. Our transformation also provides a private version of the smooth histogram framework, which is commonly used for converting streaming algorithms into sliding window variants, and achieves a multiplicative approximation to many problems, such as estimating $L_p$-norms, distinct elements, and the length of the longest increasing subsequence. 
\end{abstract}

\section{Introduction}
Approximation algorithms are often used to efficiently approximate a function $f: \mathcal{D} \rightarrow \mathbb{R}^+$ in settings where resource constraints prevent us from computing the function exactly. For example, problems such as Knapsack are $\mathsf{NP}$-$\mathsf{Hard}$ and, unless $\mathsf{P} = \mathsf{NP}$, do not admit a polynomial time solution. However, the Knapsack problem admits a fully polynomial time approximation scheme (FPTAS) i.e., for any $\alpha > 0$ there is a deterministic algorithm, running in time $\poly(n,1/\alpha)$, which outputs a solution that is guaranteed to be be nearly as good (up to multiplicative factor $1 \pm \alpha$) as the optimal solution.  As a second example, even if the problem is computationally tractable it may still be the case that the input dataset $D \in \mathcal{D}$ is extremely large, making it infeasible to load the entire dataset into RAM, or impractical to execute a linear time algorithm. To remedy such shortcomings, models such as sublinear-space and sublinear-time algorithms have been proposed. For example, one may want to estimate frequencies of elements that appear in a stream of $n$ elements up to a multiplicative $1\pm \alpha$ factor, while using only $\poly\left(\log n, \frac{1}{\alpha}\right)$ memory cells. Or, one may want to estimate the number of connected components of a dense graph on $n$ vertices up to (relative) additive error $\kappa$ by only inspecting $\poly(\log n, \frac{1}{\kappa})$ many edges. 

In addition to time and space efficiency, user privacy is another important consideration in contexts where the input to our function $f$ is sensitive user data. Differential privacy (DP)~\cite{Dwork06,DworkMNS06} is a rigorous mathematical concept that gives provable guarantees on what it means for an algorithm to preserve the privacy of individual information in the input dataset. Informally, a randomized function computed on a dataset $D$ is {\em differentially private} if the distribution of the function's output does not change significantly with the presence or absence of an individual data point. Thus, a natural goal is to develop efficient differentially private algorithms to approximate functions/queries of interest. 

One general way to preserve differential privacy is to add noise scaled to the global sensitivity $\Delta_f$ of our function $f$, i.e., the maximum amount $|f(D)-f(D')|$ that the answer could change by modifying a single record in our dataset $D$ to obtain a new dataset $D'$. This general approach yields efficient and accurate approximations for $f$ as long as we have an efficient algorithm to compute $f$ {\em exactly} and the global sensitivity of $f$ is sufficiently small. However, in some resource-constrained settings, we may need to use an approximation algorithm $\calA_f$ instead of evaluating $f$ exactly. Unfortunately, the mechanism that computes $\calA_f(D)$ and then adds noise scaled to the global sensitivity $\Delta_f$ of our function $f$ is not necessarily differentially private. In particular, even if we are guaranteed that $|\calA_f(D) - f(D)| \leq \alpha f(D)$ we might still have $|\calA_f(D)-\calA_f(D')| \geq 2 \alpha f(D) \gg \Delta_f$ for neighboring datasets $D$ and $D'$, e.g., suppose $\calA_f(D) = (1+\alpha)f(D)$ and $\calA_f(D') = (1-\alpha)f(D')$. Thus, the global sensitivity of $\calA_f$ can be quite large and adding noise proportional to $\Delta_{\calA_f}$ would prevent us from providing meaningful accuracy guarantees. This raises a natural question: Suppose that our function $f$ admits an accurate (but not necessarily resource-efficient) differentially private approximation algorithm and that $f$ also admits an efficient (but not necessarily private) approximation algorithm. Is it necessarily the case that there is also an equally efficient differentially private approximation algorithm?

Unfortunately, a result of~\cite{HardtT10, BunUV18} suggests that the answer to the previous question may be no. Suppose our dataset $D$ consists of $n$ users $x_1, \ldots, x_n$ with $n$ binary attributes i.e., $x_i \in \{0,1\}^n$. Consider the function $f(D)$ that computes all of the one-way marginals i.e., $f(D) = \langle \frac{1}{n} \sum_{i=1}^n x_i[j] \rangle_{j=1}^n \in \mathbb{R}^n$. In particular, there is a non-private sublinear time algorithm that samples $\O{\log n}$ users and (with high probability) outputs a good approximation to {\em all} $n$ one-way marginals. However, if we require that our algorithm satisfy pure, i.e, $\eps$-differential privacy (resp. approximate, i.e., $(\eps,\delta)$-differential privacy) then we need to look at {\em at least} $\Omega(n/\eps)$ (resp. $\Omega(\sqrt{n \log (1/\delta)})$) samples~\cite{HardtT10, BunUV18}. 
In light of this, we pose the following general questions: 
\begin{center}
\textit{What are sufficient conditions for an approximation algorithm to be made differentially private? \\
Can an approximation algorithm be made differentially private in an efficient black-box manner? } 
\end{center}

Over the years, many differentially private approximation algorithms have been developed for problems in optimization, machine learning, and distribution testing (see for e.g.,~\cite{GuptaLMRT10,AcharyaSZ18,Ghazi0MN21,GongXPFQ20}), in a somewhat ad-hoc manner. Often, these results give a differentially private algorithm for that specific problem and do not easily generalize to give differentially private algorithms for a large class of problems. A general framework for developing differentially private approximation algorithms for a large class of problems is desirable as this would not only make DP approximation algorithms more easily accessible to non-DP experts, but more importantly, it would shed light on what kinds of algorithms are more amenable to differential privacy. Furthermore, a framework that uses the underlying approximation algorithm as a {\em black-box} is desirable as this avoids the need to (re)design, (re)analyze, and (re)implement the new differentially private versions of these approximation algorithms. We emphasize that this type of framework has been well-studied for \emph{computing} functions privately by calibrating noise proportional to their global or smooth sensitivity~\cite{DworkMNS16, NissimRS07} (see \secref{sec:related-work} for more discussion).

Our work makes partial progress towards answering these general questions. In particular, we give an efficient black-box framework for converting a non-private approximation algorithm $\calA_f$ with tunable accuracy parameters into a differentially private approximation algorithm $\calA'_f$ with reasonable accuracy guarantees as long as the global sensitivity $\Delta_f$ of the function $f$ being approximated is sufficiently low. For the case when $\calA_f$ is deterministic, we achieve a pure $\eps$-differentially private approximation algorithm via a direct transformation, and when $\calA_f$ is randomized, i.e., has a small failure probability, we achieve $\eps$-differential privacy by first applying a transformation that gives a $(\eps,\delta)$-differentially private algorithm and then apply a postprocessing step to achieve $\eps$-differentially privacy. For example, suppose that for any $\alpha > 0$ our algorithm $\calA_f$, taking $\alpha$ and our dataset $D$ as input, provides the guarantee that $|\calA_f(D)-f(D)| \leq \alpha f(D)$ e.g., any FPTAS algorithm would satisfy our tunable accuracy requirement. In such a case, for any $\alpha >0$ we can transform our non-private algorithm $\calA_f$ into a differentially private version with multiplicative error $\alpha$ and small additive error term which (necessarily) comes from the noise that we added. Intuitively, we exploit the fact that we can run $\calA_f$ with an even smaller accuracy parameter $\rho \ll \alpha$ which can be tuned to ensure that the smooth sensitivity of our algorithm is sufficiently small. Our same general framework still applies if we allow that the approximation algorithm $\calA_f$ has a small additive error term i.e., $|\calA_f(D)-f(D)| \leq \alpha f(D) + \kappa$. If $\calA_f(D)$ is only guaranteed to output a good approximation (i.e., $|\calA_f(D)-f(D)| \leq \alpha f(D) + \kappa$) with probability $1-\delta/2$ (e.g., an FPRAS algorithm would satisfy this requirement with additive error $\kappa=0$) then our framework achieves $\eps$-differential privacy by first obtaining an approximate $(\eps, \delta)$-differential privacy algorithm and then a postprocessing step. In cases where the approximation algorithm is not tunably accurate our black-box framework does not necessarily apply\footnote{One could still apply our black-box transformation. However, the accuracy guarantees would be degraded and we would only achieve $(\eps,\delta)$-differential privacy for sufficiently large values of $\eps, \delta >0$ which depend on the approximation error parameter $\alpha$.  }. For example, the best known approximation algorithms for vertex cover achieve the guarantee $f(G) \leq \calA_f(G) \leq 2 f(G)$ i.e., because there is no way to control the smooth sensitivity of our approximation algorithm.

\subsection{Our Contributions} 
We introduce a generic black-box framework for converting certain approximation algorithms for a function $f:\mathcal{D} \rightarrow \mathbb{R}^+$ into a differentially private approximation algorithm using smooth sensitivity~\cite{NissimRS07}. We first introduce a definition for \emph{tunable} approximation algorithms used throughout our paper, and then present our main results for the DP framework, and then give new differentially private algorithms for a variety of approximation algorithms obtained via this unifying framework. 

 \begin{definition}[$(\alpha,\kappa, \delta)$-approximation]\deflab{def:akd-approx}
An algorithm $\calA_f$ is a $(\alpha,\kappa, \delta)$-approximation for $f$ if for every $D \in \calD$ with probability at least $1-\delta$, we have that $(1-\alpha)f(D) -\kappa \leq \calA_f(D) \leq (1+\alpha)f(D) + \kappa$.  
 \end{definition}

We may abuse notation and omit the failure probability $\delta$ parameter in the above definition, if it is clear from the context.  Some algorithms $\calA_f$ may take the approximation parameters $\alpha, \kappa, \delta \geq 0$ as input\footnote{We allow that $\alpha = \kappa = \delta = 0$ in which case $\calA_f$ can simply compute $f$ exactly --- whether or not this computation is efficient. }. 

\begin{definition}[tunable approximation]\deflab{def:tunable}
$\calA_f(D, \alpha,\kappa,\delta)$ provides a tunable approximation of $f$ if for  every $\alpha, \kappa, \delta \geq 0$ the algorithm $\calA_f(\cdot, \alpha,\kappa,\delta)$ obtained by hardcoding $\alpha, \kappa$ and $\delta$ is a  $(\alpha,\kappa,\delta)$-approximation for $f$. \\
When the parameters $\alpha,\kappa, \delta$ are clear from the context, we may abuse notation and just write $\calA_f(D)$. For a tunable approximation algorithm we will use $R(n,\alpha, \kappa, \delta)$ to denote the amount of a particular resource used by the algorithm. The resources we consider in this work include time, space and query complexity of the algorithm (depending on the model) which we denote by  $T(\cdot, \cdot, \cdot, \cdot)$,  $S(\cdot, \cdot, \cdot, \cdot)$ , and  $Q(\cdot, \cdot, \cdot, \cdot)$ respectively.  

\end{definition}

As a concrete example any FPTAS algorithm $\calA_f$ for $f$ would be a tunable approximation for $f$ with running time $T(n,\alpha, \kappa, \delta) = \poly(n, 1/\alpha)$ for any $\alpha >0$ and any $\kappa, \delta \geq 0$ --- an FPTAS has no additive error $(\kappa = 0)$ and zero failure probability ($\delta =0$). Similarly a FPRAS algorithm would be a tunable approximation with running time $T(n,\alpha,\kappa, \delta) = \poly(n, \alpha, \log(1/\delta))$ for any $\alpha, \delta > 0$ and any $\kappa \geq 0$ --- an FPRAS also has no additive error $(\kappa =0$). 

\paragraph{General Framework for Approximation Algorithms. }
Our main result gives a framework for converting any existing non-DP algorithm $\calA_f$ that provides an $(\alpha, \kappa,\delta)$-approximation of $f$ into an $\eps$-DP algorithm $\calA''_f$ in the following manner: (1) Apply \algref{alg:fpras-DP} to obtain an $(\eps,\delta)$-DP algorithm $\calA'_f$ that achieves an $(\alpha',\kappa',\delta')$-approximation (see~\thmref{thm:main}), (2) Apply a postprocessing step on the output of $\calA'_f$ outlined in \thmref{thm:dp-transform} to achieve an $\eps$-DP algorithm $\calA''_f$ with the same accuracy guarantees as $\calA'_f$ barring an additive error of $o(1)$. 
We emphasize that $\calA_f$ is a tunable approximation, in other words, $\calA_f$ takes the parameters $(\alpha,\kappa,\delta)$ as input.

\begin{restatable}{theorem}{thmmainfpras}\thmlab{thm:main} (($\eps, \delta$)-privacy)
Suppose that $\calA_f$ is a tunable approximation of $f: \calD \to \mathbb{R}^{+}$. Then for all $\eps > 0$, $\delta = \delta(n) > 0$\footnote{typically we set $\delta=\text{negl}(n)$ or $\delta=n^{-c}$ for some constant $c>0$. In particular $\delta(n) $ may approach zero as $n \to \infty$.}, $\alpha\geq 0$ and $ \kappa \geq 0$, there is an algorithm $\calA'_f$ such that  
\begin{enumerate}
    \item (Privacy) $\calA'_f$ is $(\eps,\delta')$-differentially private  where $\delta'=\delta(1+\exp(\eps/2))$. 
    \item (Accuracy) For all $D \in \calD$, and $0< \gamma$, with probability $1-\delta-\exp(-\gamma)$,
    \begin{align*}
    (1-\alpha')f(D) -\kappa' -  \frac{2\Delta_f}{\eps} \cdot \gamma \leq \calA'_f(D) \leq (1+\alpha')f(D) +  \kappa' + \frac{2\Delta_f}{\eps} \cdot \gamma     
    \end{align*}
        where $\alpha' = \frac{\alpha (\eps + 16\gamma)}{12 \log(4/\delta)}$, and $\kappa'= \kappa\left(\frac{2\gamma\alpha}{3\log(4/\delta)}+\frac{8\gamma}{\eps}+1\right)$, and $\Delta_f:=\max_{D,D'\in\calD,D \sim D'}\|f(D)-f(D')\|_1. $
    \item (Resource) $\calA'_f$ uses $R\left(n,\frac{\eps\alpha}{\log(4/\delta)}, \kappa, \delta \right)$ resource, where $R(\cdot, \cdot, \cdot, \cdot)$ is the resource used by $\calA_f$. 
    \end{enumerate}
\end{restatable}


We illustrate the utility of \thmref{thm:main} with specific parameters --- if we have a non-private algorithm $\calA_f$ that guarantees an $(\alpha,0,\delta)$-approximation, then for constant $\eps$, $\delta=\frac{1}{n^c}$ and $\gamma=c\log(n)$, we see that the DP algorithm $\calA_f$ achieves an $\left(\alpha(1+o(1)), \O{\frac{\Delta_f\log(n)}{\eps}}, \frac{1}{n^c}\right)$-approximation. 
We typically use these parameters for $\delta,\gamma$ in our applications for streaming and sublinear-time algorithms. 

Our reduction in \thmref{thm:main} is quite simple -- we describe the associated \algref{alg:fpras-DP} below. 

\begin{algorithm}[!htb]
\caption{$(\eps,\delta)$-differentially private framework $\calA'_f$ for tunable approximation algorithm $\calA_f$}
\alglab{alg:fpras-DP}
\begin{algorithmic}[1]
\Require{Input set $D$, accuracy parameters $\alpha \in (0,1)$ and $\kappa$, DP parameter $\eps$, DP failure probability $\delta\in(0,1)$, approx. algorithm $\calA_f$. }
\State Let $x_A:=\calA_f(D, \rho,\tau, \delta/2)$, where $\rho:=\left(\frac{\eps  \alpha }{12 \log(4/\delta)}\right)$, and $\tau:=\kappa$. 
\State{\Return $x_A+X$ where $X \sim \Lap  \left(\frac{2(4 \rho x_A+4\tau+\Delta_f)}{\eps} \right)$}
\end{algorithmic}
\end{algorithm}
Note that in~\algref{alg:fpras-DP}, we leave our additive parameter $\kappa$ as is when running $\calA_f$, but we still choose to define $\tau:= \kappa$. This is because depending on the problem, and the accuracy/efficiency guarantees desired, we can set $\tau$ to be a tuned version of $\kappa$ (for e.g., we set $\tau:= \kappa/\log(n)$ for the problem of estimating the number of connected components).

\begin{remark}\remlab{rem:median-trick}
We also note that, even if the failure probability $\delta >0$ of $\calA_f$ is non-negligible, that we can always boost the success probability by running $\calA_f(D)$ multiple times and computing the median over all outputs. Even if the error rate $0 < \delta < 1/2$ is a constant we can always reduce the failure probability to a lower target $0 < \delta' \ll \delta$ while increasing the running time by a multiplicative factor $O\left(\log(1/\delta') \right)$. In particular, we can set $\delta'$ to be a negligible function of $n$ such as $\delta' = n^{-\log n}$ whilst only incurring a $\O{\log^2 n}$ blowup in our running time. 
\end{remark}

We stress that we can only apply~\thmref{thm:main} to existing non-DP algorithms $\calA_f$ that give an approximation guarantee of the form $(1-\alpha)f(D) -\kappa \leq \calA_f \leq (1+\alpha)f(D)+\kappa$. 
For example, we cannot apply \thmref{thm:main} to obtain an $(\eps,\delta)$-DP algorithm for estimating the minimum vertex cover size in sublinear time. This is because the non-DP sublinear-time algorithm $\calA_{vc}$ has an approximation guarantee of the form $2VC(G) -\kappa n \leq \calA_{vc} \leq 2VC(G)+\kappa n$. On the other hand, we can use our DP framework to obtain an $(\eps,\delta)$-DP algorithm for obtaining a $(0,\kappa n,\delta)$-approximation of the maximum matching size in sublinear time (see~\corref{cor:max-match}). Intriguingly, both the minimum vertex cover size and the maximum matching size algorithms use the same underlying strategy of estimating a greedy maximal matching in a local fashion, but since they return different estimators based on the objective and we can only use our framework as a black-box, we cannot apply our framework to the former while we can still apply it to the latter. 

Finally, by applying a post-processing step described below, we show how to obtain an $\eps$-DP algorithm from the $(\eps,\delta)$-DP algorithm obtained in \thmref{thm:main}. Importantly, the accuracy guarantee of the resulting $\eps$-DP algorithm only differs by a small additive factor of $1/(KM)$, where $M = \max_D f(D)$ is the maximum possible output value, e.g., $M \leq n^3$ for triangle counting and $K>0$. Moreover for negligible $\delta$, the accuracy guarantee of the resulting pure DP algorithm still holds with high probability. 

\begin{restatable}{theorem}{dptransform}\thmlab{thm:dp-transform}
   Let $M = \max_D f(D)$ and let parameter $K>0$. If $\calA'_f(D)$ is $(\eps,\delta)$-DP algorithm with accuracy guarantee $(1-\alpha) f(D) - \kappa \leq \calA_f(D) \leq (1+\alpha) f(D) + \kappa$ holding with probability $1-\eta$ then there exists an algorithm $\calA''_f(D)$ which is $\eps$-DP with accuracy guarantee $(1-\alpha) f(D) - \kappa - \frac{1}{KM} \leq \calA_f(D) \leq (1+\alpha) f(D) + \kappa + \frac{1}{KM}$ with probability at least $1-\eta - p$ where $p = \frac{\delta K(M+1)}{e^\eps - 1 +\delta K(M+1)}$.   
\end{restatable}

Our second result is an analogous framework for converting any existing deterministic non-DP approximation algorithm $\calA_f$ that provides an $(\alpha, \kappa,0)$-approximation of $f$ into an $\eps$-DP algorithm $\calA'_f$.

\begin{restatable}{theorem}{thmmainfptas}
\thmlab{thm:main-2} ($\eps$-privacy)
Suppose that $\calA_f$ is a deterministic tunable approximation of $f: \calD \to \mathbb{R}^{+}$.
Then for all $\eps > 0$, $\alpha\geq 0$ and $ \kappa \geq 0$, there is an algorithm $\calA'_f$ such that
\begin{enumerate}
    \item (Privacy) $\calA'_f$ is $\eps$-differentially private. 
    \item (Accuracy) For all $D \in \calD$, we have that with probability $\geq 9/10$,
    $$(1-\alpha')f(D) - \kappa' - \frac{7 \Delta_f }{\eps} \leq \calA'_f(D) \leq (1+\alpha')f(D) + \kappa' + \frac{7 \Delta_f }{\eps}$$
    where $\alpha':=\alpha C_1 (\eps+C_2\gamma)$, $\kappa':=\kappa C_3 (\alpha+ \frac{C_4}{\eps})$ for some constants $C_1,C_2,C_3,C_4>0$ and $\Delta_f:=\max_{D,D'\in\calD,D \sim D'}\|f(D)-f(D')\|_1. $
    \item (Resource) $\calA'_f$ uses $R\left(n,\frac{\eps\alpha}{36}, \kappa\right)$ resource, where $R(\cdot, \cdot, \cdot)$ is the resource used by $\calA_f$. 
\end{enumerate}
\end{restatable}


\paragraph{DP Sublinear-time Results. }
We use~\thmref{thm:main} in conjunction with \thmref{thm:dp-transform} in a black-box manner to obtain pure differentially-private sublinear time algorithms for several problems (see~\tableref{tab:sublinear} for a summary).

 In many models of sublinear-time computation the efficiency of the algorithm is measured in the number of queries made to the input, rather than the time complexity of the algorithm. It is often the case that the two are polynomially related, but there are instances in which the actual time complexity of the algorithm may be exponentially larger than the query complexity, in terms of the approximation factor. Nevertheless, in these instances too, the literature uses time and query complexity interchangeably. This is because the sublinear-time model assumes restricted or expensive access to the input, while further computation on local machines with the answers obtained from queries is considered to be cheap. We use query complexity for the sake of clarity.

We note that in the sublinear-time literature, the approximation parameters $\alpha,\kappa$ are usually considered to be a constant, but the analyses for most of these theorems hold for  $\alpha=\alpha(n),\kappa=\kappa(n) \in (0,1)$, where $n$ is the input size.

Here we do not explicitly define the sublinear model (or the queries allowed) for each problem, see \secref{sec:sublinear} for more details. For a graph $G$ we denote the number of vertices as $n$, the number of edges as $m$, and the average degree of the graph as $\bar{d}$. 

Typically, the accuracy guarantees of the non-DP results are presented with probability at least 2/3 --- in order to apply our framework, we apply the median trick (see~\remref{rem:median-trick}) to boost the probability of success to $1-\delta$. For simplicity of comparing our results, for any constant $c>0$, we set $\delta:= 1/n^c$ in the sequel.

\begin{table}[!ht] 
\centering
\begin{adjustbox}{width=1\textwidth,center=\textwidth}
\small
\begin{tabular}{|c|c|c|c|c|c|}\hline 
 Problem & Reference & Privacy & Mult. error & Add. error & Query Complexity \\ \hline
 \multirow{2}{*}{Number of Triangles} &\cite{EdenLRS17} & Non-Private & $\alpha$ & 0 &    $\O{(\frac{n}{t^{1/3}} + \frac{m^{3/2}}{t})\poly(\log(n),\frac{1}{\alpha})}$\\\cline{2-6} 
 &This Work & $\eps$-edge DP &  $\alpha$ & $\O{\frac{n \log(n)}{\eps}}$ &    $\O{(\frac{n}{t^{1/3}} + \frac{m^{3/2}}{t})\poly(\log(n),\frac{1}{\alpha\eps})}$ \\\hline 
\multirow{2}{*}{Connected Components} &\cite{BerenbrinkKM14} & Non-Private & 0 & $\kappa n$ &   $\O{\frac{1}{\kappa^2} \log\left(\frac{1}{\kappa}\right)\log(n) }$ \\\cline{2-6} 
 &This Work & $\eps$-edge DP & 0 &  $\O{\kappa n} +\O{\frac{\log(n)}{\eps}}    $&   $\O{\frac{\log^3(n)}{\kappa^2} \log\left(\frac{\log(n)}{\kappa} \right)}$\\\hline 
\multirow{2}{*}{Weighted MST}  & \cite{ChazelleRT05}  & Non-Private & $\alpha$ & 0 & $\O{\bar{d}w\alpha^{-2}\log\left(\frac{\bar{d}w}{\alpha}\right)\log(n)}$  \\\cline{2-6}
  & This Work  & $\eps$-edge DP  & $\alpha$ &$\O{\frac{\log(n)}{\eps}}$  & $\O{\bar{d}w\frac{\log^2(n)}{\alpha^2\eps^2}\log\left(\frac{\bar{d}w\log(n)}{\alpha\eps}\right)\log(n)}$ \\\hline
\multirow{4}{*}{Average Degree}  & \cite{GoldreichR04} & Non-Private & $\alpha$ & 0 &  $\O{\frac{n}{\sqrt{m}} \poly\left(\frac{\log(n)}{\alpha}\right)\log(n)}$\\\cline{2-6}
&\multirow{2}{*}{\cite{BlockiGM22}}  & \multirow{2}{*}{$\eps$-edge DP} & \multirow{2}{*}{$\alpha$} & \multirow{2}{*}{0} &   $\O{\sqrt{n} \poly\left(\frac{\log(n)}{\alpha}\right)\poly\left(\frac{1}{\eps}\right) \log(n)}$\\
& &  &  &  & (analysis assumes $\bar{d}\geq 1$)   \\\cline{2-6}
 & \multirow{2}{*}{This Work} & \multirow{2}{*}{$\eps$-edge DP} & \multirow{2}{*}{$\alpha$} & \multirow{2}{*}{0} &  $\O{\frac{n}{\sqrt{m}} \poly\left(\frac{\log^2(n)}{\alpha \eps} \right)\log(n)}$\\
 & &  &  &  & for $\bar{d}= \Omega({\frac{\log(n)}{n\eps}})$\\ \hline 
\multirow{2}{*}{Maximum Matching Size}  & \cite{YoshidaYI12} & Non-Private & 0 & $\kappa n$ &  $\O{d^{\O{1/\kappa^2}}\log(n)}$ \\\cline{2-6}
& This Work & $\eps$-node DP &  0 & $\O{\frac{\kappa n}{\eps}}$ &  $\O{d^{\O{1/\kappa^2}}\log(n)}$ \\\hline
\multirow{2}{*}{Distance to Bipartiteness}  & \cite{GhoshMRS22} & Non-Private & 0 & $\kappa n^2$ &  $\O{(1/\kappa^3)\log(n)}$ \\\cline{2-6}
& This Work & $\eps$-edge DP &  0 & $\O{\kappa n^2}+\O{\frac{\log(n)}{\eps}}$ &  $\O{(\log^4(n)/\kappa^3)}$ \\\hline
\end{tabular}
\end{adjustbox}
\caption{Summary of Sublinear-time DP graph algorithms obtained via our black-box DP transformation. According to our notation multiplicative error $\alpha$ means a multiplicative factor of $(1\pm \alpha)$. See~\secref{sec:sublinear} for details. } 
\tablelab{tab:sublinear}
\end{table}

We give the first (to the best of our knowledge) $\eps$-DP sublinear time algorithm for estimating the number of triangles, connected components, and the weight of a minimum spanning tree whose accuracy guarantees hold with high probability.

For estimating the average degree of a graph, in recent work,~\cite{BlockiGM22} gave a pure $\eps$-DP algorithm that achieves an $(\alpha,0)$-approximation --- a crucial observation is that their analysis only holds under the assumption that the average degree is at least one i.e., $\bar{d} \geq 1$ (see~\secref{sec:sublinear} for details). 
In this work, we remove the need for this assumption in the DP setting, by directly applying our black-box DP transformation to the original algorithm of~\cite{GoldreichR04} which works substantially better whenever we have $m = \omega(n)$ edges. 

For estimating the maximum matching size in a graph, although~\cite{BlockiGM22} gave an $\eps$-DP algorithm for estimating the maximum matching size that achieves a 2-multiplicative factor and $\kappa n$ additive factor, 
they left the task of finding an $(0,\kappa n)$-approximation in the DP setting as an open problem. In this work, we partially resolve this problem by presenting an $\eps$-DP algorithm that gives a  $(0,\O{\frac{\kappa n}{\eps}})$-approximation of the maximum matching size. 
Crucially, our resulting analysis cannot guarantee that the added Laplace noise will be small with high probability, but only guarantees this will be the case with \emph{constant probability}. This problem highlights a limitation of our black-box DP framework --- if the non-DP algorithm that we want to apply our DP transformation on has a time/space/query complexity that has an exponential dependence on the approximation parameters then the resulting DP algorithm that achieves a similar approximation guarantee with high probability may be highly inefficient in terms of time/space/query complexity. 

We also show how to apply our DP framework to an  algorithm estimating the distance to bipartiteness in dense graphs \cite{GhoshMRS22, AlonVKK03}, which is accurate with probability $1-o(1).$ The same reduction can be similarly applied to other natural properties that enjoy the feature that they admit distance-estimation algorithms with $\poly(1/\kappa)$ query complexity, where $\kappa$ is the additive (normalized) error. For example, in the fundamental results of \cite{GoldreichGR98} an efficient distance approximation algorithm for the maximum $k$-cut problem, and thus $k$-colorability is presented.  \cite{FiatR21}, also based on results from \cite{AnderssonE02}, generalizes these properties to the notion of ``semi-homogeneous partition properties" and show efficient distance estimation algorithms for properties such as Induced $P_3$-freeness, induced $P_4$-freeness,  and chordality.\footnote{In general, distance estimation is closely related to tolerant testing \cite{ParnasRR06}, and for dense graph properties it is known that if a property is testable with a number of queries of the form $f(\kappa)$, then they admit a distance estimator \cite{FischerN07} with an exponential blowup in $\frac{1}{\kappa}$   in the query complexity. Hence, in its general form the query complexity of estimating the distance to ``hereditary" graph properties is a tower of exponential of height $\poly(1/\kappa)$ \cite{AlonFNS09}. }

\paragraph{DP Streaming Results. }
We also apply our framework given by~\thmref{thm:main} and \thmref{thm:dp-transform} to obtain differentially-private streaming algorithms for many fundamental problems, i.e., see \tableref{tab:streaming} and \tableref{tab:sliding}. 
We remark that while the accuracy guarantees of our resulting algorithms may be surpassed by recent works studying these problems on an individual basis, our applications are black-box reductions that avoid individual utility and privacy analysis of each non-private streaming algorithm, which can be heavily involved and quite non-trivial, e.g.,~\cite{MirMNW11,BlockiBDS12,Smith0T20,BuGKLST21,WangPS22,BravermanMWZ22}. 

In the streaming model, elements of an underlying dataset arrive one-by-one and the goal is to compute or approximate some predetermined function on the dataset using space that is sublinear in the size of the dataset. 
Our reductions also have wide applications to various archetypes of data stream models, which we now discuss. 
In insertion-only streams, the updates of the stream increment the underlying dataset, such as adding edges to a graph, adding terms to a sequence, or increasing the coordinates of a frequency vector. 
In turnstile (or dynamic) streams, the updates of the stream can both increase and decrease (or insert and delete) elements of the underlying dataset. 
Finally, in the sliding window model, only the $W$ most recent updates of the data stream define the underlying dataset. 
Both the turnstile streaming model and the sliding window model are generalizations of insertion-only streams, and our framework has implications in all three models. 

We first show that our framework can be applied to existing non-private dynamic algorithms for weighted minimum spanning tree, $L_p$ norm estimation for $p\ge 1$ (and also $F_p$ moment estimation for $0<p<1$), and distinct elements estimation. 
Thus using our framework, we essentially get private dynamic algorithms for these problems for free (in terms of correctness, not optimality). 
Since the dynamic streaming model generalizes the insertion-only streaming model, we also obtain private streaming algorithms in the insertion-only model as well. 
We summarize these results in \tableref{tab:streaming}.

\begin{table}[!htb] 
\centering
\begin{adjustbox}{width=1\textwidth,center=\textwidth}
\small
\begin{tabular}{|c|c|c|c|c|c|}\hline 
 Problem & Reference & Privacy & Mult. error & Add. error & Space Complexity\\ \hline
\multirow{2}{*}{Weighted Minimum Spanning Tree} &\cite{AhnGM12}& Non-Private & $\alpha$ & 0 &   $\O{\frac{1}{\alpha}n\log^4 n}$ \\\cline{2-6} 
 &This Work & $\eps$-DP & $\alpha$ & $\O{\frac{M\log m}{\eps}}$ &   $\O{\frac{1}{\alpha\eps} n \log^5 n}$\\\hline 
\multirow{2}{*}{$L_p$-norm, $p>2$} &\cite{GangulyW18}& Non-Private & $\alpha$ & 0 &   $\O{\frac{1}{\alpha^2} n^{1-\frac{2}{p}}\log^2 n + \frac{1}{\alpha^{4/p}} n^{1-\frac{2}{p}} \log^{2/p} n \log^2 n}$ \\\cline{2-6} 
 &This Work & $\eps$-DP & $\alpha$ & $\O{\frac{\log m}{\eps}}$  &  $\O{\frac{p}{ \alpha^2\eps^2} n^{1-2/p}}\cdot\poly\left(\log n,\log\frac{1}{\alpha\eps}\right)$\\\hline 
 \multirow{2}{*}{$L_p$-norm, $p=2$} &\cite{AlonMS99} & Non-Private & $\alpha$ & 0 &   $\O{\frac{1}{\alpha^2}\log^2 n}$ \\\cline{2-6} 
 &This Work & $\eps$-DP & $\alpha$ & $\O{\frac{\log m}{\eps}}$  &  $\O{\frac{1}{\alpha^2\eps^2}\log^4 n}$ \\\hline 
 \multirow{2}{*}{$L_p$-norm, $p\in (0,2)$} &\cite{KaneNPW11}& Non-Private & $\alpha$ & 0 &   $\O{\frac{1}{\alpha^2}\log^2 n}$   \\\cline{2-6} 
 &This Work & $\eps$-DP & $\alpha$ & $\O{\frac{\log m}{\eps}}$  &  $\O{\frac{1}{\alpha^2\eps^2}\log^4 n}$  \\\hline 
  \multirow{2}{*}{$L_p$-norm, $p=0$} & \cite{KaneNW10}&Non-Private & $\alpha$ & 0 &  $\O{\frac{1}{\alpha^2}\log^2 n\log\frac{1}{\alpha}}$  \\\cline{2-6} 
 &This Work & $\eps$-DP & $\alpha$ & $\O{\frac{\log m}{\eps}}$  &  $\tO{\frac{1}{\alpha^2\eps^2}\log^4 n}$ \\\hline 
\end{tabular}
\end{adjustbox}
\caption{Summary of DP algorithms in the dynamic/turnstile model obtained via our black-box DP transformation. According to our notation multiplicative error $\alpha$ means a multiplicative factor of $(1\pm \alpha)$. See \secref{sec:dynamic} for details.  }\tablelab{tab:streaming}
\end{table}

We then apply our framework in~\thmref{thm:main} to the sliding window model. 
To that end, we first recall that given a $(\alpha,0)$-approximation algorithm for the insertion-only streaming model, the smooth histogram framework~\cite{BravermanO10} provides a transformation that obtains a $(\alpha,0)$-approximation algorithm in the sliding window model for a ``smooth'' function. 
Although there are problems that are known to not be smooth, e.g.,~\cite{BravermanOZ12,BravermanDMMUWZ20,BravermanWZ21,EpastoMMZ22,JayaramWZ22}, the smooth histogram framework does provide a $(\alpha,0)$-approximation to many important problems, such as counting, longest increasing subsequence, $L_p$ norm estimation for $p\ge 1$ (and also $F_p$ moment estimation for $0<p<1$), and distinct elements estimation. 
We remark that if we tried to apply the non-private smooth histogram framework to a DP insertion-only streaming algorithm, this might preserve privacy by post-processing, but may significantly increase the error in terms of accuracy.
On the other hand, our framework avoids these issues and achieves private analogs of these algorithms in the sliding window model without compromising utility. 
We summarize our results for the sliding window model in \tableref{tab:sliding}. We note that in recent work,~\cite{EpastoMMMVZ23} give a generalized smooth histogram approach to convert a DP continual release streaming algorithm into a sliding window algorithm in the continual release setting. We focus on the one-shot streaming setting in our work.

\begin{table}[!htb] 
\centering
\begin{adjustbox}{width=1\textwidth,center=\textwidth}
\small
\begin{tabular}{|c|c|c|c|c|c|}\hline 
 Problem & Reference & Privacy & Mult. error & Add. error & Space Complexity\\ \hline
\multirow{2}{*}{Longest Increasing Subsequence} &\cite{SunW07}& Non-Private & $\alpha$ & 0 &   $\O{\frac{k^2}{\alpha}\log^2 n}$ \\\cline{2-6} 
 &This Work & $\eps$-DP & $\alpha$ & $\O{\frac{\log m}{\eps}}$ &  $\O{\frac{k^2}{\alpha\eps} \log^4 n}$\\\hline 
\multirow{2}{*}{Distinct Elements} &\cite{Blasiok20}& Non-Private & $\alpha$ & 0 &   $\O{\frac{1}{\alpha^3}\log^2 n}$ \\\cline{2-6} 
 &This Work & $\eps$-DP & $\alpha$ & $\O{\frac{\log m}{\eps}}$  &  $\O{\frac{1}{\alpha^3\eps^3}\log^5 n}$\\\hline 
 \multirow{2}{*}{$L_p$-norm, $p=2$} &\cite{WoodruffZ21} & Non-Private & $\alpha$ & 0 &   $\O{\frac{1}{\alpha^2}\log^3 n\log^3\frac{1}{\alpha}}$ \\\cline{2-6} 
 &This Work & $\eps$-DP & $\alpha$ & $\O{\frac{\log m}{\eps}}$  &  $\tO{\frac{1}{\alpha^2\eps^2}\log^5 n\log^3\frac{1}{\alpha\eps}}$  \\\hline 
 \multirow{2}{*}{$L_p$-norm, $p\in (0,2)$} &\cite{WoodruffZ21}& Non-Private & $\alpha$ & 0 &   $\O{\frac{1}{\alpha^2}\log^3 n(\log\log n)^2\log^3\frac{1}{\alpha}}$   \\\cline{2-6} 
 &This Work & $\eps$-DP & $\alpha$ & $\O{\frac{\log m}{\eps}}$  &  $\tO{\frac{1}{\alpha^2\eps^2}\log^5 n}$  \\\hline 
\end{tabular}
\end{adjustbox}
\caption{Summary of DP algorithms in the sliding window model obtained via our black-box DP transformation. According to our notation multiplicative error $\alpha$ means a multiplicative factor of $(1\pm \alpha)$. 
See \secref{sec:sw} for details.}
\tablelab{tab:sliding}
\end{table}

\subsection{Our Techniques} 

Given a tunable $(\alpha,\kappa,\delta)$-approximation algorithm $\calA_f$ for the function $f:\calD \to \mathbb{R}^+$, our goal is to obtain a differentially private approximation algorithm that achieves a target $(\alpha',\kappa',\delta')$-approximation of $f$ where $\alpha',\kappa'$ are in terms of $\alpha,\kappa$.

\paragraph{Warm-up: When $\calA_f$ is deterministic and only has multiplicative error. }
For simplicity, let us first consider an $(\alpha,0,0)$-approximation algorithm $\calA_f$, in other words, $\calA_f$ \emph{always} outputs a value such that $(1-\alpha)f(D) \leq \calA_f(D) \leq (1+\alpha)f(D)$. Since we want to make $\calA_f$ differentially private, intuitively, we need to add noise to the output of $\calA_f$. The local sensitivity of $\calA_f$ at $D$ (i.e., $LS_{\calA_f}(D) = \max_{D' \sim D} |\calA_f(D)-\calA_f(D')|$) is upper bounded by $2\alpha f(D) + \Delta_f$. Since $\Delta_f$ is small and we can tune $\alpha$ to be arbitrarily small, it is tempting to think that we can just add noise proportional to $2\alpha f(D) + \Delta_f$. Unfortunately, scaling noise proportional to local sensitivity is not necessarily private. On the other hand we could ensure privacy by scaling noise proportional to the global sensitivity (i.e., $\max_{D \in \calD} LS_{A_f}(D) \leq \max_{D \in \calD} 2\alpha f(D) + \Delta_f$) but noise will likely be too large to obtain meaningful accuracy guarantees.  We adopt the strategy of adding noise proportional to the smooth sensitivity~\cite{NissimRS07} of $\calA_f$ instead. In particular,~\cite{NissimRS07} observed that if we can find a ``sufficiently smooth" function $S_f(D) \geq LS_{\calA_f}(D)$ upper bounding the local sensitivity of $\calA_f$ then we can preserve privacy by computing $\calA_f(D)$ and adding noise scaled according to $S_f(D)$. 

We can show that the function $S_f(D) = 4 \alpha \calA_f(D) + \Delta_f$ is a $\beta$-smooth upper bound on the local sensitivity of $\calA_f$ for $\beta=6\alpha$ where $S_f$ is $\beta$-smooth if $S_f(D) \leq e^{\beta} S_f(D')$ for all pairs of neighboring datasets $D \sim D'$. To achieve privacy using the smooth sensitivity framework we need to ensure that $\beta$ is sufficiently small relative to our privacy parameters $\eps$ and $\delta$ (if applicable). For example, we can achieve $\left(\eps,\delta \left(1+\exp{\left(\frac{\eps}{2}\right)}\right)\right)$-differential privacy by adding Laplace Noise scaled by $\frac{2 S_f(D)}{\eps}$, but only if $S_f$ is $\beta$-smooth for $\beta \leq \frac{\eps}{2 \ln (2/\delta)}$. For pure differential privacy we require that $\beta < \frac{\eps}{2(\lambda +1)}$ where $\lambda$ is a parameter of the noise distribution --- smaller $\lambda$ implies higher variance.

If we want to ensure that the output is accurate, we also need to ensure that the calibrated noise with $S_f(D)$ is small e.g., $o(f(D)) + \O{\Delta_f}$. Note that by definition since $S_f(D) = 4 \alpha \calA_f(D) + \Delta_f$, and we add noise proportional $S_f(D)$, we expect that the noise added may be $>\alpha f(D)$. Thus, in order to address this challenge, our basic strategy is to run the original (non-private) approximation algorithm $\calA_f$ with \emph{tuned} error factors e.g., we decrease $\alpha$ by a multiplicative factor of $\frac{\eps}{\ln(n)}$, let $\rho:= \frac{\eps \alpha}{\ln(n)}$. Since we are now running $\calA_f(D,\rho,0,0)$, we have that the function $S_f(D) = 4 \rho \calA_f(D) + \Delta_f$ is a $\beta$-smooth upper bound on the local sensitivity of the algorithm $\calA_f(\cdot, \rho, 0,0)$. Assuming the global sensitivity $\Delta_f$ is small, we can now show that w.h.p. the noise sampled proportional to $S_f(D)$ is at most $\alpha f(D) + O\left(\Delta_f/\eps \right)$ thus resulting in an $\eps$-differentially private algorithm with reasonable accuracy. 

By tuning the parameter $\alpha$ we actually accomplish two useful properties (1. accuracy) we decrease both the local sensitivity and our smooth upper bound $S_f(D)$ which reduces the magnitude of the noise that we add, and (2. privacy) we achieve $\beta$-smoothness for increasingly small values of $\beta$ so that the required condition $\beta \leq \frac{\eps}{2 \ln (2/\delta)}$ (or $\beta < \frac{\eps}{2(\lambda +1)}$) can be satisfied if we want to scale noise according to $S_f(D)$. 

\paragraph{Extending to deterministic $\calA_f$ with multiplicative and additive error. }More generally, if we have an $(\alpha,\kappa,0)$-approximation algorithm $\calA_f$ then we can show that $S_f(D) = 4 \alpha \calA_f(D) + \Delta_f + 4\tau$ is a $\beta$-smooth upper bound on the local sensitivity of $\calA_f$ with $\beta = 6\alpha$ (see~\lemref{lem:smooth-bound-fptas}). In particular, note that the additive error term $\kappa$ does not adversely impact smoothness. Thus, we can achieve pure differentially privacy by tuning $\alpha$ such that $6 \alpha < \beta < \frac{\eps}{2(\lambda +1)}$ and scaling our noise according to $S_f(D)$~(see~\lemref{lem:dp-fptas}). We can also obtain stronger accuracy guarantees by relaxing the requirement for pure DP and tuning $\alpha$ such that $6\alpha \leq \beta \leq \frac{\eps}{2 \ln (2/\delta)}$ so that we can sample our noise from the Laplace distribution which has strong concentration guarantees.

\paragraph{When $\calA_f$ is randomized. }The remaining challenge is to handle randomized approximation algorithms $\calA_f$ which are only guaranteed to output a good approximation with high probability i.e., with non-zero probability $\delta >0$ the algorithm is allowed to output an arbitrarily bad approximation. In particular, let us consider an $(\alpha,\kappa,\delta)$-approximation algorithm $\calA_f$. For any possible input $D $ we are always guaranteed that with probability $\geq 1-\delta$  the algorithm $\calA_f(D)$ outputs a good approximation $(1-\alpha)f(D) \leq \calA_f(D) \leq (1+\alpha)f(D)$. Unfortunately, the function $S_f(D) = 4 \alpha \calA_f(D)  + \Delta_f + 4\kappa$ is no longer guaranteed to be a $\beta$-smooth upper bound on the local sensitivity of $\calA_f$ since $\calA_f$ may sometimes output a value outside the specified approximation bounds. 

In order to address this challenge, we define a function $g_f(D)$ that matches $\calA_f(D)$ with probability at least $1-\delta/2$ and is {\em always} guaranteed to output a good approximation. We emphasize that $g_f(D)$ may not be efficiently computable, but it is well-defined and only used for the purpose of analysis. More specifically, we set $g_f(D) = \calA_f(D)$ as long as $(1-\alpha)f(D) - \kappa  \leq \calA_f(D) \leq (1+\alpha)f(D) + \kappa$. If $\calA_f(D)>(1+\alpha)f(D)$, then we define $g_f(D):= (1+\alpha)f(D)$, similarly, if $\calA_f(D)<(1-\alpha)f(D)$, then we define $g_f(D):= (1-\alpha)f(D) + \kappa$. Observe that we are always guaranteed that $(1-\alpha) f(D) - \kappa \leq g_f(D) \leq (1+\alpha)f(D) + \kappa$. Thus, $S_f(D) = 4\alpha g_f(D) + \Delta_f + 4 \tau$ is a $\beta=6\alpha$-smooth upper bound on the local sensitivity of $g_f$ (see~\lemref{lem:smooth-bound}). As long as $ 6\alpha \leq \beta \leq \frac{\eps}{2 \ln (2/\delta)}$ we could preserve $\left(\eps,\delta\left(1+\exp{\left(\frac{\eps}{2}\right)}\right)\right)$-differential privacy by  outputting $g_f(D)$ plus  Laplace Noise scaled by $\frac{8\alpha g_f(D) + \Delta_f + 8 \tau}{\eps}$ i.e., scaled according to our $\beta$-smooth upper bound on the local sensitivity of $g_f$. Unfortunately, the function $g_f$ may not be efficiently computable. Thus, we substitute $g_f$ for $\calA_f$ and instead output $\calA_f(D)$ plus Laplace noise scaled according to $\frac{8\alpha \calA_f(D) + \Delta_f + 8 \tau}{\eps}$. While $4\alpha \calA_f(D) + \Delta_f + 4\tau$ is not necessarily a $\beta$-smooth upper bound on the local sensitivity of $\calA$, the key point is that the latter (efficiently computable) procedure is equivalent to the former (differentially private) procedure as long as $g_f(D) = A_f(D)$ which happens as long as $\calA_f$ outputs a good approximation i.e., except with probability $\delta/2$. Thus, we can apply a hybrid argument to argue that the final efficiently computable algorithm is $\left(\eps,\frac{\delta}{2} + \delta \left(1+\exp{\left(\frac{\eps}{2}\right)}\right)\right)$-differential privacy (see~\lemref{lem:dp-fpras}). In order to ensure accuracy, we use the same strategy as before, i.e., we run $\calA_f(D,\rho,\tau,\delta/2)$, where $\rho \leq \frac{\eps \alpha }{\log(1/\delta)}$. Sampling noise proportional to $S_f(D)$ (where $S_f(D)$ is now defined in terms of $\rho$), and absorbing the failure probability of algorithm $\calA_f$ into the DP failure probability term $\delta$, results in an approximate differentially private algorithm. Finally applying the postprocessing step results in a pure differentially private algorithm. We refer to the full proofs (~\secref{sec:gen-transform}) for additional details. 

\paragraph{Applications. }We give some intuition on how we apply~\thmref{thm:main} to various applications by choosing appropriate parameters. Recall that with probability $1-\delta-\exp(-\gamma)$, $\calA'_f$ outputs $(1-\alpha')f(D) -\kappa' -  \frac{2\Delta_f}{\eps} \cdot \gamma \leq \calA'_f(D) \leq (1+\alpha')f(D) +  \kappa' + \frac{2\Delta_f}{\eps} \cdot \gamma $, where $\alpha' = \frac{\alpha (\eps + 16\gamma)}{12 \log(4/\delta)}$, and $\kappa'= \kappa\left(\frac{2\gamma\alpha}{3\log(4/\delta)}+\frac{8\gamma}{\eps}+1\right)$ with a time/space/query complexity blow-up incurred by running the original algorithm $\calA_f$ with multiplicative accuracy parameter $\rho = \frac{\eps\alpha}{\log(4/\delta)}$. First, observe that if the original algorithm $\calA_f$ has time/space/query complexity with a dependence on $\poly(\frac{1}{\alpha})$, then the resulting time/space/query complexities for $\calA'_f$ will still have a polynomial dependence, i.e., $\poly(\frac{\log(4/\delta)}{\alpha})$ --- this naturally leads to FPRAS or FPTAS applications, as well as other classes of approximation algorithms like sublinear time or space. On the otherhand, if the time/space/query complexity of $\calA_f$ has a non-polynomial dependence on $1/\alpha$, e.g., $\exp(\frac{1}{\alpha})$, then since $\delta$ is typically $\negl(n)$ or $\frac{1}{n^c}$ for $c>0$, the resulting DP algorithm $\calA'_f$ could have much worse time/space/query-guarantees with respect to $n$, e.g., in an extreme case if we set $\delta=2^{-\poly(n)}$ then $\rho = \Omega( \poly(n)/\alpha)$ and we could incur a $\exp(\frac{\poly(n)}{\alpha})$ multiplicative overhead in the running time. It is worth noting that one could optionally reduce the additive error term $\kappa'$ for $\calA_f'$ by reducing the error term $\kappa$ for $\calA$.  

We further emphasize this trade-off between obtaining small failure probability bounds and the accuracy or resource guarantees. Consider the following two examples---  (Example 1) if we set the probability of failure, i.e., $\exp(-\gamma) = \delta = \frac{1}{n^c}$ for any $c>0$, then the resulting approximation parameters are roughly $\alpha' = \alpha(1+o(1))$, and $\kappa'= \kappa(\alpha + \frac{\log(n)}{\eps}+1)$\footnote{In the applications we consider the original (non-private) approximation algorithm typically has only multiplicative or only additive error and not both. In particular, we typically either have $\alpha > 0$ and $\kappa =0$ or $\kappa>0$ and $\alpha =0$, but not the case where $\alpha > 0$ and $\kappa > 0$.  Considering the case when $\kappa \neq 0$, and $\alpha=0$, we (roughly) have $\kappa' = \kappa(\frac{\log(n)}{\eps}+1)$.}, and the additional error term depending on global sensitivity is roughly $\frac{\Delta_f \log(n)}{\eps}$. We incur a time/space overhead by running $\calA_f$ with multiplicative accuracy parameter $\rho = \Omega(\eps \alpha/\log n)$ instead of $\alpha$. (Example 2) if we set the probability of failure, i.e., $\exp(-\gamma) = \delta = \frac{1}{n^{\log\log(n)}} - \negl(n)$, then $\alpha'$ remains the same as before, and now $\kappa'= \kappa(\alpha + \frac{\log(n)\log\log(n)}{\eps}+1)$, but the additional error term depending on global sensitivity becomes $\frac{\Delta_f \log(n)\log\log(n)}{\eps}$. In this latter case we incur time or space overhead by running $\calA_f$ with multiplicative accuracy parameter $\rho = \Omega\left(\frac{\eps \alpha}{\log n \log \log n}\right)$ --- to reduce the $\kappa' \log n \log \log n$ error term it could be useful to run $\calA_f$ with additive error parameter $\frac{\kappa}{\log n \log \log n}$ which may incur additional time or space overhead. 
Thus these two examples illustrate how, as we decrease the failure probability, the accuracy and resource (time/space in this case) guarantees become worse. See \secref{sec:sublinear}, and \secref{sec:streaming} for applications to sublinear time and streaming algorithms.

\subsection{Related Work}\seclab{sec:related-work}
One of the first differentially private frameworks for computing general functions was introduced by~\cite{DworkMNS16} which released functions with additive noise, where the noise is calibrated according to the \emph{global sensitivity} of the function $f$. This framework was generalized by~\cite{NissimRS07}, to handle functions which might have a high global sensitivity but are usually less sensitive in practice. The framework allows the release of functions with instance-specific noise, where the noise that is added is not just determined by $f$ but by the input dataset as well. The noise magnitude calibrated is according to the \emph{smooth sensitivity} of $f$ on the input dataset which is a smooth upper bound on the \emph{local sensitivity} of $f$ on an input dataset. The smooth sensitivity of a function may be hard to compute, therefore in the same work,~\cite{NissimRS07} give a generic method called the \emph{sample and aggregate} method that bypasses the explicit computation of the smooth sensitivity of the function and works even when the function is given as a black-box. \cite{DworkL09} suggested a framework called \emph{Propose-Test-Release} to release statistical estimators with additive noise where the noise is calibrated according to the \emph{local sensitivity} of the estimator. Note that adding noise proportional to the local sensitivity of a function with respect to an input set usually does not preserve privacy, but their approach first proposes a bound on the local sensitivity and privately tests whether this bound holds for the specific input set, and then releases the noisy response to the query. 

In the context of developing differentially private frameworks for approximation algorithms, \cite{BlockiGM22} formally introduced the notion of \emph{coupled global sensitivity} of a randomized algorithm, which gives an analogous framework as that of the global sensitivity framework~\cite{DworkMNS16}, but for randomized approximation algorithms instead of deterministic functions. In this framework, one can run a non-private randomized approximation algorithm $\calA_f(D)$ on the dataset, and privacy is obtained by adding noise proportional to the coupled global sensitivity of $\calA_f$. More formally, the coupled global sensitivity measures the worst-case $L_1$-sensitivity of the outputs of a randomized algorithm $\calA_f$ on neighboring inputs over a minimum coupling of the internal coin tosses of $\calA_f$.  

In independent work, Tetek~\cite{Tetek22} also explores the problem of transforming randomized approximation algorithms into (pure) differentially private approximation algorithms. In contrast to our results Tetek's transformation \cite{Tetek22} assumes that the error of the original approximation algorithm either has small subexponential diameter or bounded mean error --- assumptions that would not apply generically to every (tunable) approximation algorithm. Assuming subexponential error their work shows that it is possible to achieve $\eps$-DP by adding Laplace Noise yielding  accuracy guarantees that hold with high probability. However, the assumption of the error being subexponential is quite strong and does not often hold for many randomized approximation algorithms. While assuming bounded mean error is a weaker assumption on the error of the non-private randomized algorithm, however the DP noise is sampled from the Pareto distribution, which has polynomial tail bounds. This leads to accuracy guarantees which only hold with constant probability. Note that applying the median trick commonly used to amplify success probability in the non-private literature adversely affects the privacy budget and is thus not desirable. In contrast, our transformation applies generically to any (tunably) accurate approximation algorithm and we achieve accuracy guarantees that hold with high probability for the same problems studied in their paper. Finally, we correct an outdated claim\footnote{A prior version of the paper achieved pure DP, but that transformation (Theorem 1.5)  only applied to deterministic tunable approximation algorithms} from the comparison to our work detailed in~\cite{Tetek22} that says that we only achieve approximate privacy. We can achieve {\em pure} DP algorithms by applying a postprocessing step to the output of our transformation as outlined in~\thmref{thm:dp-transform}.

\section{Preliminaries}

{We use the notation $\tO{f(n)}$ to mean $f(n)\cdot\polylog(f(n))$.}
We define datasets $D$ and $D'$ as \emph{neighboring}, denoted as $D\sim D'$, if removing or adding one point in $D$ results in $D'$; alternatively, if changing one data point in $D$ results in $D'$.
\begin{definition}[Differential privacy]\label{def:DP}
\cite{DworkMNS06}
An algorithm $\calA$ is $(\eps,\delta)$-DP if for every pair of neighboring datasets $D \sim D'$, and for all sets $\calS$ of possible outputs, we have that $\Pr[\calA(D) \in \calS ] \leq e^\eps \Pr [\calA(D') \in \calS] + \delta $. When $\delta=0$ we simply say that the algorithm is $\eps$-DP.
\end{definition}
Given an $(\eps,\delta)$-DP algorithm, one can obtain an $\eps$-DP algorithm under certain conditions outlined below. We include the proof for completeness in \appref{epsdp-eps}.
\begin{restatable}{theorem}{thmapproxtopure}[Approximate DP to Pure DP]
\thmlab{thm:epsdel-eps}
Let $\calA: \calD \to \calR$. If $\calA$ is an $(\eps,\delta)$-DP algorithm such that $\delta \leq \frac{(e^\eps -1)p}{\vert \calR \vert (1-p)}$ then there is an algorithm $\calA'$ such that $\calA'$ is $\eps$-DP defined in the following manner.
\begin{align*}
    \calA'(D) = \begin{cases} 
    \calA(D) & \text{with probability }1-p \\
    \mathsf{random}(\calR) &\text{with probability }p
    \end{cases}
\end{align*}
where $\calR$ is the range of $\calA_f$. 
\end{restatable}

We define the distributions we will use to sample additive noise from below. 
\begin{definition}[Laplace distribution]
We say a random variable $X$ is drawn from a Laplace distribution with mean $\mu$ and scale $b>0$ if the probability density function of $X$ at $x$ is $\frac{1}{2b}\exp\left(-\frac{|x-\mu|}{b}\right)$. 
We use the notation $X\sim\Lap(b)$ to denote that $X$ is drawn from the Laplace distribution with scale $b$ and mean $\mu=0$. 
\end{definition}
\begin{definition}[Cauchy distribution]
We say a random variable $X$ is drawn from a Cauchy distribution with location parameter $x_0$ and scale $b>0$ if the probability density function of $X$ at $x$ is $\frac{1}{\pi b}\left(\frac{b^2}{(x-x_0)^2+b^2}\right)$. 
We use the notation $X\sim\C(b)$ to denote that $X$ is drawn from the Cauchy distribution with scale $b$ and location parameter $x_0=0$. 
\end{definition}

We formally define the concept of global sensitivity which is a worst-case notion of sensitivity for deterministic functions below. 
\begin{definition}[Global sensitivity]
The global sensitivity of a function $f:\calD\to\mathbb{R}^d$ is defined by
\[\Delta_f=\max_{D,D'\in\calD,D \sim D'}\|f(D)-f(D')\|_1.\]
\end{definition}

We define the notion of local sensitivity for a fixed input, which can be much smaller than the global sensitivity, but in general, adding noise calibrated according to the local sensitivity does not preserve DP. 

\begin{definition}[Local sensitivity]
For $f:\calD\to\mathbb{R}$ and $D\in\calD$, the local sensitivity of $f$ at $D$ is defined as
\[LS_f(D)=\max_{D':D\sim D'} \|f(D)-f(D')\|_1.\]
Note: if $f:\calD \times \mathcal{R} \to \mathbb{R}$ is a randomized function which, in addition to a dataset $D \in \calD$ takes random coins $r \in \mathcal{R}$ as input we simply define $LS_f(D) =  \max_{r \in \mathcal{R}} LS_{f_r}$ where $f_r(D) \doteq f(D;r)$.
\end{definition}

In order to add instance-specific noise, we define the notions of $\beta$-smooth upper bound which is a smooth upper bound on the local sensitivity. 

\begin{definition}[Smooth upper bound on local sensitivity]\deflab{def:smooth-ub}
For $\beta>0$, a function $S:\calD\to\mathbb{R}$ is a $\beta$-smooth upper bound on the local sensitivity of $f:\calD\to\mathbb{R}$ if
\begin{enumerate}
\item \label{it:smooth-1}
For all $D\in \calD$, we have $S(D)\ge LS_f(D)$.
\item \label{it:smooth-2}
For all $D,D'\in \calD$ with $\|D-D'\|_1=1$, we have $S(D)\le e^\beta\cdot S(D')$. 
\end{enumerate}
\end{definition} 
Finally, although one cannot add noise calibrated with local sensitivity, one can add noise proportional to a $\beta$-smooth upper bound on the local sensitivity as follows. 
\begin{theorem}[Corollary 2.4 in \cite{NissimRS07}]
\thmlab{thm:smooth:laplace}
Let $f:\calD \to\mathbb{R}$ and $S:\calD\to\mathbb{R}$ be a $\beta$-smooth upper bound on the local sensitivity of $f$. 
\begin{enumerate}
    \item If $\beta \leq \frac{\eps}{2(\lambda+1)}$ and $\lambda>1$, the algorithm $D \to f(D)+\frac{2(\lambda+1)S(D)}{\eps} \cdot \eta$, where $\eta$ is sampled from the distribution with density $h(z) \propto \frac{1}{1+\vert z \vert^{\lambda}}$, is $\eps$-differentially private. 
    \item If $\beta\le\frac{\eps}{2\ln(2/\delta)}$ and $\delta\in(0,1)$, then the algorithm $D \to f(D)+\frac{2S(D)}{\eps}\cdot \eta$ where $\eta \sim \Lap(1)$ is $(\eps,\delta')$-differentially private for $\delta' = {\delta}\left(1+\exp\left(\frac{\eps}{2}\right)\right)$\footnote{These bounds differ slightly from those listed in the original paper (Corollary 2.4 in \cite{NissimRS07}). We confirmed with the authors in private communication that $\delta$ should be multiplied by $(1+\exp(\eps/2))$.}. 
\end{enumerate}
\end{theorem}

\section{General Transformation for Approximation Algorithms}\seclab{sec:gen-transform}
In this section, we formally define our black-box differentially private transformation for (randomized) approximation algorithms. Given a tunable approximation (see \defref{def:tunable}) algorithm of $f$, call it $\calA_f$, that outputs an $(\alpha,\kappa,\delta)$-approximation, our framework for randomized algorithms involves two steps --- (1) Apply \algref{alg:fpras-DP} to $\calA_f$ to obtain an $(\eps,\delta)$-DP algorithm $\calA'_f$ with accuracy guarantees outlined in \thmref{thm:main} (2) Apply postprocessing step to the output of $\calA'_f$ to obtain an $\eps$-DP algorithm (see \thmref{thm:dp-transform}).

We first prove \thmref{thm:main} that provides theoretical guarantees for algorithm $\calA'_f$ (\algref{alg:fpras-DP}). This is our main contribution as the postprocessing step to obtain pure DP applies a folkore result. 

Observe that even for the case when the original algorithm $\calA_f$ gives an $(\alpha,0,\delta)$-approximation of $f$ (i.e., $\kappa=0$), the resulting DP algorithm $\calA'_f$ will still have an additive error, this additive error is inherent due to the requirement of adding Laplace noise to preserve DP. We emphasize that the Laplace noise added to the output of algorithm $\calA_f$ depends on the global sensitivity of the function $f$, therefore, we can only get meaningful DP approximation algorithms using this transformation for functions with low global sensitivity.

\thmmainfpras*

\begin{proof} $\calA_f'$ is defined in \algref{alg:fpras-DP} --- it first runs $\calA_f(D,\rho,\kappa, \delta/2)$ where $\rho := \frac{\eps \alpha}{12 \log(4/\delta)}$ and then adds Laplace Noise. Thus, the resource used by $\calA'_f$ is $R\left(n, \rho, \kappa,\delta \right)$. The privacy guarantee follows from \lemref{lem:dp-fpras}, and the accuracy guarantee follows from \lemref{lem:dp-acc-fpras}.
\end{proof}

\begin{remark}
When $\calA_f$ is a PRAS, by definition, the output of $\calA_f$ is an $(\alpha,0,\delta)$-approximation of $f$ running in time $T(n, \alpha, 0, \delta) = \poly(n, 1/\alpha, \log(1/\delta))$. Applying \thmref{thm:main} with negligible $\delta =n^{-\log n}$ and $\gamma = \log^2 n$ for any $\alpha' >0$ we obtain a private $\left(\alpha', O\left(\frac{\Delta_f }{\eps \log^2 n}\right) ,2n^{-\log n}\right)$-approximation with polynomial running time $\poly(n, 1/\eps, 1/\alpha')$.
\end{remark}

\begin{lemma}
\lemlab{lem:smooth-bound}
Let $0<\rho<1/2$. Suppose that $\calA_f$ outputs a $(\rho, \tau,\delta)$-approximation of a function $f:\calD \to \mathbb{R}^+$ with global sensitivity $\Delta_f$. Let $\calA_{f,R}$ denote a deterministic run of $\calA$ using a fixed set of random coins $R$. Define function $g_{f,R}$ by 
\begin{equation*}
g_{f,R}(D)=
    \begin{cases}
        \calA_{f,R}(D) &\text{if } (1-\rho)f(D) - \tau\leq \calA_{f,R}(D) \leq (1+\rho)f(D)+ \tau\\
        (1-\rho)f(D)-\tau  &\text{if } \calA_{f,R}(D) <(1-\rho)f(D) - \tau\\
        (1+\rho)f(D)+\tau &\text{if } \calA_{f,R}(D) > (1+\rho)f(D) +\tau\\
    \end{cases}
\end{equation*}
Then the function $S_f(D)= 4\rho g_{f,R}(D)+4\tau+\Delta_f$ is a  $\beta$-smooth upper bound for $g_{f,R}$ where $\beta \geq 6\rho $.
\end{lemma}

\begin{proof}
Fix an arbitrary set of random coin tosses $R$. We frequently use the fact that $(1-\rho)f(D)-\tau  \leq g_{f,R}(D) \leq (1+\rho)f(D)+\tau$. We also note that since $0\leq  \rho < 1/2$, we have that $\frac{1}{2}f(D) - \tau \leq g_{f,R}(D) \leq 2 f(D) + \tau$.

First, we show that Condition~\ref{it:smooth-1} of \defref{def:smooth-ub} holds. Without loss of generality, we assume $f(D)\geq f(D')$, where $D'$ is any neighboring input. 
Then: 
\begin{align*}
LS_{g_{f,R}}(D) &= \max_{D':D\sim D'} \|g_{f,R}(D)-g_{f,R}(D')\| &\\
&\leq \|(1+\rho)f(D)+\tau-(1-\rho)f(D')+\tau\| &\\
&\leq \rho \|f(D)+f(D') \|+2\tau+\Delta_f  &\\
&\leq 2\rho f(D)+2\tau + \Delta_f &\\
&\leq 4\rho g_{f,R}(D)+2\rho\tau+2\tau+\Delta_f &\text{as }\frac{1}{2}f(D) -\tau \leq g_{f,R}(D) \\
&\leq 4\rho g_{f,R}(D)+4\tau+\Delta_f &\\
&= S_f(D) &
\end{align*}
Next, we show that Condition~\ref{it:smooth-2} of \defref{def:smooth-ub} holds below. 
We have:

\begin{align*}
    S_f(D) &=  4 \rho g_{f,R}(D)+4 \tau+\Delta_f &\\
    &\leq 4\rho \left(1+\rho\right)f(D)+4\rho\tau+4\tau+\Delta_f  & \text{by def of }g_{f,R}\\
    &\leq 4\rho \left(1+\rho \right)(\Delta_f + f(D'))+4\rho\tau+4\tau+\Delta_f  & \text{since }D\sim D'\\
    &\leq 4 \rho(1+\rho) f(D') +(4\rho(1+\rho)+1)\Delta_f+4\rho\tau+4\tau & \\
    &\leq 4\rho \frac{(1+\rho)}{1-\rho}(g_{f,R}(D')+\tau) +(1+6\rho)\Delta_f+4\rho\tau+4\tau& \text{by def of }g_{f,R}\\
    &\leq 4\rho (1+\rho)(1+2\rho)(g_{f,R}(D')+\tau)+(1+6\rho)\Delta_f+4\rho\tau+4\tau& \\
    &\leq 4\rho (1+\rho)(1+2\rho)g_{f,R}(D')+12\rho\tau+(1+6\rho)\Delta_f+4\rho\tau+4\tau& \\
    &\leq 4\rho (1+\rho)(1+2\rho) g_{f,R}(D')+(1+6\rho)(\Delta_f+4\tau)& \\
    &\leq 4\rho (1+4\rho) g_{f,R}(D')+(1+6\rho)(\Delta_f+4\tau) & \\
    &\leq (1+6\rho)(4\rho g_{f,R}(D')+\Delta_f+4\tau) & \\
    &\leq  e^{6\rho} \cdot (4\rho g_{f,R}(D') +\Delta_f+4\tau)  = e^{\beta} S_f(D'),
\end{align*}
where $\beta \geq 6 \rho$.

\end{proof}

\begin{remark} As a special case if we have a $(0,\kappa, 0)$-approximation algorithm $\calA_f$ (i.e., no multiplicative error, zero failure probability) then applying \lemref{lem:smooth-bound} yields the smooth upper bound $S_f(D)= 4\tau+\Delta_f$. We observe that this smooth upper bound is independent of $D$ and, therefore, $S_f$ is just an upper bound on the global sensitivity of $g_f$. Furthermore, in this special case we are guaranteed that $\calA_f(D) = g_f(D)$ with probability $1$. Thus, in this special case, we can achieve pure $\eps$-DP by computing $\calA_f(D)$ and adding Laplace noise proportional to $S_f$.    \\
If we have a $(0,\kappa, \delta)$-approximation algorithm for $\delta \neq 0$ then we still have $S_f(D)= 4\tau+\Delta_f$ which means that $S_f(D)$ is an upper bound on the global sensitivity of $g_f$. However, computing $\calA_f(D)$ and adding Laplace noise proportional to $S_f$ does not necessarily yield a pure DP algorithm since we may have $\Delta_f(D) \neq \calA_f(D)$ with non-zero probability $\delta$. If we have $(\alpha,\kappa, 0)$-approximation algorithm $\calA_f$, and $\alpha\neq 0$, since $S_f(D)=4 \rho \calA_f(D)+4\tau + \Delta_f$ still depends on the input $D$. However, we can still achieve DP using \thmref{thm:main-2}. 
\end{remark}

Applying \lemref{lem:smooth-bound} to the theorem calibrating noise to smooth bounds on the smooth sensitivity~\cite{NissimRS07} we show that the \algref{alg:fpras-DP} preserves privacy below. 

\begin{lemma}[Privacy]\lemlab{lem:dp-fpras}
\algref{alg:fpras-DP} is $(\eps,\delta')$-differentially private where $\delta'=\delta(1+\exp(\eps/2))$.
\end{lemma}

\begin{proof} Consider a modification of \algref{alg:fpras-DP}, call it \algref{alg:fpras-DP}' where instead of computing $\calA_f(D,\rho,\tau,\delta/2)$ we instead sample the random coins $R$ that $\calA_f$ would have used and replace the value $\calA_f(D,\rho,\tau,\delta/2; R)$ (which we denote as $\calA_{f,R}(D)$ in the sequel) with $g_{f,R}(D)$.
The function $g_{f,R}$ may not be efficiently computable, but we only use \algref{alg:fpras-DP}' for the purpose of analysis. We first observe that by \lemref{lem:smooth-bound}, for any $\beta \geq 6 \rho$ the function $S_f(D) = 4\rho g_{f,R}(D) + 4\tau+\Delta_f$ is a $\beta$-smooth upper bound on the sensitivity of $g_{f,R}$. 
Thus, by \thmref{thm:smooth:laplace}, it is sufficient to set $\rho:=\frac{\eps \alpha }{12 \log(4/\delta)}$ for $6 \rho \leq \beta\leq \frac{\eps}{2\ln\left(\frac{4}{\delta}\right)}$ and add noise proportional to $\Lap \left(\frac{2S_f(D)}{\eps} \right)= \Lap \left(\frac{2(4 \rho g_{f,R}(D)+4\tau+\Delta_f)}{\eps} \right)$ to preserve $(\eps,\delta/2)$-privacy of \algref{alg:fpras-DP}'.

Since $g_{f,R}(D)$ is $\calA_{f,R}(D)$ except with probability $\delta/2$, \algref{alg:fpras-DP} is identical to \algref{alg:fpras-DP}' except with probability $\delta/2$. Thus, this shows that \algref{alg:fpras-DP} is $(\eps,\delta)$-private.  
\end{proof}

 \begin{fact}\factlab{fact:lap}
 If $Y \sim \Lap(b)$, then  $\Pr[ \vert Y \vert \geq \ell \cdot b] = \exp(-\ell)$.
 \end{fact}

\begin{lemma}[Accuracy]\lemlab{lem:dp-acc-fpras}

For all $\gamma>0$, with probability $1-\exp(-\gamma) - \delta$, 
\begin{align*} 
&\bigl(1- \rho(1+\frac{16\gamma}{\eps})\bigr)f(D) -\tau\bigl(\frac{8\gamma\rho}{\eps}+\frac{8\gamma}{\eps}+1\bigr) -\frac{2\Delta_f \gamma}{\eps} \leq  \calA'(D) \leq \bigl(1+ \rho(1+\frac{16\gamma}{\eps})\bigr)f(D) \\&+ \tau\bigl(\frac{8\gamma\rho}{\eps}+\frac{8\gamma}{\eps}+1\bigr) +\frac{2\Delta_f \gamma}{\eps}.
\end{align*}
\end{lemma}
\begin{proof}

First, using \factref{fact:lap}, for any $\gamma>0$, we have that, 
\begin{align*}
 \PPr{\vert X \vert \geq \frac{2(4 \rho \calA(D)+4\tau+\Delta_f)}{\eps} \cdot \gamma} &= \exp(-\gamma)  
\end{align*}
$\calA_f$ is a $(\rho,\tau,\delta/2)$-approximation of $f$ so for any $D \in \calD$, we have that $\calA_f(D) \leq (1+\rho) f(D) + \tau $ with probability $1-\delta/2$. Since $0 < \rho < 1/2$ we have $ (1+\rho) f(D) + \tau  \leq 2f(D)+\tau$.
Therefore, by a union bound,  
\begin{align*}
    &\Pr\bigl[\bigl(\calA_f(D) > (1+\rho) f(D) + \tau \bigr) \vee \bigl(\calA_f(D) < (1-\rho) f(D) - \tau \bigr) \vee  \bigl( \vert X \vert \\&\geq \frac{2(4 \rho \calA_f(D)+4\tau+\Delta_f)}{\eps} \cdot \gamma \bigr)\bigr] \leq \delta/2+\exp(-\gamma)  
\end{align*}
Thus, with probability $1-\exp(-\gamma)-\delta/2$, we have that 
\begin{align}\label{eq:af}
(1-\rho)f(D) - \tau \leq \calA_f(D) \leq (1+\rho) f(D) + \tau \leq 2f(D) + \tau    
\end{align}
{and} 
\begin{align}\label{eq:lapx}
\vert X \vert < \frac{2(4 \rho \calA_f(D)+4\tau+\Delta_f)}{\eps} \cdot \gamma    
\end{align}
By plugging in Eq.~\ref{eq:af} into Eq.~\ref{eq:lapx}, we have that with probability $1-\exp(-\gamma)-\delta/2$,
\begin{align*}
\vert X \vert < \frac{2(4 \rho (2f(D)+\tau)+4\tau+\Delta_f)}{\eps} \cdot \gamma    
\end{align*}
Overall, this means that with probability $1-\exp(-\gamma)-\delta/2$,
\begin{align*}
    &(1-\rho)f(D) -  \tau- \frac{2(4 \rho (2f(D)+\tau)+4\tau+\Delta_f)}{\eps} \cdot \gamma  \leq  \calA'_f(D) \\&\leq (1+\rho)f(D) + \tau+ \frac{2(4 \rho (2f(D)+\tau)+4\tau+\Delta_f)}{\eps} \cdot \gamma
\end{align*}
Grouping the like terms together gives the theorem statement. 
 \end{proof}

\dptransform*

The proof of \thmref{thm:dp-transform} is in \appref{epsdp-eps}. At a high level our idea is to define $\calA'_f(D) = \frac{\lceil K \calA_f(D) \rceil }{KM}$. The rounding step introduces a small additive error term $\leq \frac{1}{KM}$ and ensures that $\calA'_f(D)$ now has bounded range $\vert \calR\vert \leq (M+1)K$. Since $\calA'_f$ has bounded range we can apply a folklore result (see~\thmref{thm:epsdel-eps}) to transform this $(\eps,\delta)$-DP algorithm into an $\eps$-DP algorithm. 

\subsection{Achieving Pure DP for Approximation Algorithms with Zero Failure Probability}  \seclab{sec:approx-zero-delta}

In this section we show how one can achieve pure differential privacy ($\delta = 0$) when we have a tunable $(\alpha,\kappa, 0)$-approximation algorithm. The basic framework is the same except that we use the Cauchy distribution instead of Laplace when applying the Smooth Sensitivity framework --- see \thmref{thm:smooth:laplace}. Since we assume $\delta=0$ in this section we will sometimes simplify notation and write $T(n,\alpha,\kappa)$ (resp. $S(n,\alpha,\kappa)$) instead of $ T(n,\alpha,\kappa, 0) $ (resp. $S(n,\alpha,\kappa,0)$).

\begin{algorithm}[!htb]
\caption{$\eps$-differentially private framework for tunable deterministic approximation algorithms}
\alglab{alg:fptas-DP}
\begin{algorithmic}[1]
\Require{Input set $D$, accuracy parameter $\alpha \in (0,1)$, differential privacy parameter $\eps$, approx. algorithm $\calA_f$. }
\State{Let $x_A:= \calA_f(D,\rho,\tau,0)$ where $\rho:= \frac{\eps \alpha}{36}$} and $\tau:=\kappa$.
\State{\Return $x_A+X$ where $X \sim \C\left(\frac{6(4\rho x_A +\Delta_f)}{\eps}\right)$} 
\end{algorithmic}
\end{algorithm}

\begin{remark}
When $\calA_f$ is a PTAS, by definition, the output of $\calA_f$ is an $(\alpha,0,0)$-approximation of $f$ running in time $T(n, \alpha, 0) = \poly(n, 1/\alpha)$. Applying \thmref{thm:main-2} with for any $\alpha >0$ we obtain a private $\left(\alpha, O\left(\frac{\Delta_f }{\eps }\right) ,9/10 \right)$-approximation with polynomial running time $\poly(n, 1/\eps, 1/\alpha)$.
\end{remark}

\thmmainfptas*
\begin{proof}
$\calA_f'$ is defined in \algref{alg:fptas-DP} --- we first run $\calA_f(D,\rho,\kappa)$ where $\rho := \frac{\eps\alpha}{36}$ and then we add noise proportional to the standard Cauchy distribution. Thus, the resource used will be $R(n,\rho,\kappa)$.

The privacy guarantee follows from \lemref{lem:dp-fptas}, and the accuracy guarantee follows from \lemref{lem:dp-acc-fptas}. 
\end{proof}

\begin{lemma}
\lemlab{lem:smooth-bound-fptas}
Suppose that $\calA_f$ outputs a $(\rho,\tau,0)$-approximation where $0<\rho < 1/2$ of a function $f:\calD \to \mathbb{R}^+$ with global sensitivity $\Delta_f$. 
Then the function $S_f(D)= 4\rho \calA_f(D)+4\tau +\Delta_f$ is a  $\beta$-smooth upper bound for $\calA_f$ where $\beta \geq 6\rho $.

\end{lemma}
The proof remains the same as in \lemref{lem:smooth-bound}. Applying \lemref{lem:smooth-bound-fptas} to the theorem calibrating noise to smooth bounds on the smooth sensitivity~\cite{NissimRS07} we show that the \algref{alg:fptas-DP} preserves privacy below.

\begin{lemma}\lemlab{lem:dp-fptas}
\algref{alg:fptas-DP} is $\eps$-differentially private.
\end{lemma}

\begin{proof}
We first observe that by \lemref{lem:smooth-bound-fptas}, $S_f(D)= 4\calA_f(D)+4\tau + \Delta_f$ is a  $\beta$-smooth upper bound for $\calA_f$. Recall that $\rho:=\frac{\eps \alpha }{36}$, thus we can apply \thmref{thm:smooth:laplace} (with $\lambda=2$) where $6\rho \leq \beta\leq \frac{\eps}{6}$ and conclude that it is sufficient to add noise proportional to $\C\left(\frac{2(2+1)S_f(x)}{\eps}\right)=\C\left(\frac{6(4\rho \calA_f(D)+4\tau +\Delta_f)}{\eps}\right)$ to preserve $\eps$-privacy.

\end{proof}

 \begin{fact}\factlab{fact:cauch}
 If $Y \sim \C(x;0,b)$, then  $\Pr[ \vert Y \vert \geq \ell b] = 1-\frac{2\tan^{-1}(\ell)}{\pi}$.
 \end{fact} 
 
\begin{lemma}\lemlab{lem:dp-acc-fptas}
For all $\gamma>6.5$, with probability at least $9/10$, 
\begin{align*}
&\left(1-\rho \left(1+ \frac{48 \gamma }{\eps}\right)\right)f(D) - \frac{24  (\rho+1) \gamma \tau}{\eps} - \frac{6\Delta_f}{\eps}\cdot \gamma \leq \calA'_f(D) \\&\leq  \left(1+\rho \left(1+ \frac{48 \gamma }{\eps}\right)\right)f(D) + \frac{24  (\rho+1) \gamma \tau}{\eps}+\frac{6\Delta_f}{\eps}\cdot \gamma 
\end{align*}
\end{lemma}
\begin{proof}
First, we invoke \factref{fact:cauch} below,
\begin{align*}
    \Pr\left[\vert X \vert \geq  \frac{6(4\rho \calA_f(D)+4\tau+\Delta_f)}{\eps} \cdot \gamma \right] = 1 - \frac{2\tan^{-1}(\gamma)}{\pi} \leq \frac{1}{10}
\end{align*}
where the final inequality comes from using the fact that $\gamma>6.5$. 
In other words, with probability $\geq 9/10$, 
\begin{align}\label{eq:x-det}
    \vert X \vert \leq  \frac{6(4\rho \calA_f(D)+4\tau+\Delta_f)}{\eps} \cdot \gamma 
\end{align}
$\calA_f$ is a $(\rho,\tau,0)$-approximation of $f$ so for any $D \in \calD$, we have that $\calA_f(D) \leq (1+\rho) f(D) + \tau $. Since $0 < \rho < 1/2$ we have $ (1+\rho) f(D) + \tau  \leq 2f(D)+\tau$.

By plugging in the relation $\calA_f(D) \leq 2f(D) + \tau $ into Eq.~\ref{eq:x-det}, we have that with probability at least $9/10$, $$\vert X \vert \leq  \frac{6(8\rho f(D)+4\rho\tau+4\tau+ \Delta_f)}{\eps} \cdot \gamma$$ 
Thus with probability at least 9/10, 
\begin{align*}
    &(1-\rho)f(D)-\tau - \frac{6(8\rho f(D)+4\rho\tau+4\tau+ \Delta_f)}{\eps} \cdot \gamma \leq \calA'_f(D) \\&\leq  (1+\rho)f(D)+\tau + \frac{6(8\rho f(D)+4\rho\tau+4\tau+ \Delta_f)}{\eps} \cdot \gamma
\end{align*}
Rearranging the like terms together in the above expression completes the proof. 
\end{proof}

\begin{remark}
For simplicity, we have chosen to sample from the standard cauchy distribution $\lambda=2$, more generally, if we sample noise with density $h(z) \propto \frac{1}{1+\vert z \vert^{\lambda}}$, where $\lambda=c$, then with probability $1-\delta$, $\gamma=\frac{1}{\delta^{1/c}}$ in \lemref{lem:dp-acc-fptas}. 
\end{remark}

\paragraph{Application to the Knapsack Problem}
As a fun example we consider the knapsack problem. The knapsack problem is well known to be NP-Hard, but also admits an FPTAS. To define an instance of the knapsack problem we have a maximum weight capacity $W$ for the knapsack and $n$ items each with a value $v_{max} \geq v_i \geq 0$ and a weight $w_i \geq 0$. The goal is to find a subset $S \subseteq [n]$ of items to put in the knapsack maximizing the total value $v(S) = \sum_{i \in S} v_i$ subject to the constraint that the total weight $w(S) = \sum_{i \in S} w_i$ does not exceed our capacity i.e., $w(S) \leq W$. 

For the purpose of this illustration let's fix the capacity $W$ and weights $w_1,\ldots, w_n$ and let $f(v_1,\ldots,v_n)$ denote the value of the optimal knapsack solution given values $v_1,\ldots, v_n$. Let's say that two knapsack instances $(W,v_1,\ldots,v_n, w_1,\ldots,w_n)$ and $(W, v_1',\ldots, v_n', w_1,\ldots, w_n)$ are neighbors if $\sum_i \|v_i-v_i'| \leq 1$. Thus, we are viewing the exact value of each item as sensitive and the goal of differential privacy is to prevent an attacker from inferring these sensitive values exactly. Observe that the global sensitivity of $f$ is upper bounded by $\Delta_f \leq \max_{v \sim v'} \max_{S \subseteq [n]} \left|v(S)-v'(S) \right| \leq 1$\footnote{We could also define neighboring knapsack instances such that we can completely replace the value of any item i.e.,  $v$ and $v'$ are neighbors if there exists some index $i \in [n]$ such that $v_i \neq v_i'$ and $v_j = v_j'$ for all $j \neq i $. However, in this case we can we would have large global sensitivity $\Delta_f =v_{max}$. Thus, we won't be able to design an accurate differentially private approximation even if we are willing to solve the NP-Hard knapsack problem exactly. }. 

Since there is an FPTAS algorithm for Knapsack we can find a non-private approximation algorithm $\calA_f(\vec{v}, \alpha, \kappa = 0)$ running in time $T(n,\alpha) = \poly(n,1/\alpha)$. If we apply \thmref{thm:main-2}  then for any target $\alpha'$ our $\eps$-DP algorithm $\calA_f'$ runs in time $\poly(n,1/\eps,1/\alpha)$ and solves Knapsack with additive error $\O{1/\eps}$ and multiplicative error $\alpha'$ with probability at least $9/10$. If we don't require pure DP then we can also apply \thmref{thm:main} then for any target $\alpha'$ our algorithm $\calA_f'$ runs in time $\poly(n,1/\eps, 1/\alpha, \log(1/\delta))$ and solves Knapsack with probability at least $1-\delta - \exp(-\gamma)$ with additive error at most $\O{\gamma/\eps}$ and multiplicative error $\alpha'$.

\section{Private Sublinear-time Algorithms}\seclab{sec:sublinear}
We present a variety of DP results for sublinear-time graph algorithms by directly applying our DP black-box transformation (see~\thmref{thm:main} and \thmref{thm:dp-transform}). Let $\mathcal{G}_n$ denote the set of all $n$-node graphs. For a graph $G \in \mathcal{G}_n$ on $m$ edges, we denote $\bar{d}$ as the average degree of $G$. Graphs $G_1=(V, E_1)$, $G_2=(V, E_2)$ are \emph{node-neighboring}, denoted by $G_1\sim_v G_2$, if there exists a vertex $v\in V$ such that $E_1(V\setminus\{v\})=E_2(V\setminus\{v\})$. Graphs $G_1$ and $G_2$ are \emph{edge-neighboring} i.e., $G_1 \sim_e G_2$ if there exists an edge $e$ such that $E_1 \setminus{\{e\}}=E_2 \setminus{\{e\}}$.

We first define the sublinear models that we will be working with in the following problems. In the \emph{adjacency list query} model the following queries can be answered in constant time --- (1) degree queries: given $v \in V$, output $\deg(v)$. (2) neighbor queries: given $v \in V$, and $i \in [n]$, output the $i$-th neighbor of $v$ or $\perp$ if $i>\deg(v)$. In the \emph{adjacency matrix query} model the following queries can be answering in constant time--- pair queries: given $u,v \in V$, output whether $(u,v)$ is an edge in $G$ or not. The \emph{general query} model allows for degree queries, neighbor queries, and pair queries. 

Typically the success probability for non-DP sublinear-time algorithms is stated in terms of constants (i.e., $\geq 2/3$). We implicitly apply the standard median trick to boost the probability of success to $1-\delta$ through $\O{\log\frac{1}{\delta}}$ parallel iterations, and state the non-DP results with success probability $1-\delta$ in the following exposition. 
\paragraph{Estimating the number of Triangles. }Given a graph $G$, we study the problem of estimating the number of triangles in the general query model. \cite{EdenLRS17} gave an approximation algorithm for this problem with an expected query complexity of $\O{(\frac{n}{t^{1/3}} + \frac{m^{3/2}}{t})\poly(\log(n),1/\alpha)\log(1/\delta)}$. 
In order to achieve a high probability bound on the query-complexity, we run $\Theta(\log(1/\delta))$ instances of the algorithm in parallel and return the output of the instance that terminates first. 

\begin{theorem}~\cite{EdenLRS17}
Let $t$ be the number of triangles in graph $G$. For any $0<\alpha<1$, there is an algorithm that makes $\O{(\frac{n}{t^{1/3}} + \frac{m^{3/2}}{t})\poly(\log(n),1/\alpha)\log(1/\delta)}$ queries and outputs an estimate $\hat{t}$ such that with probability at least $1-\delta$, 
$$ (1-\alpha) \cdot t \leq \hat{t} \leq (1+\alpha)\cdot t $$
\end{theorem}

Observe that the global sensitivity (under edge-DP) of the function that computes the number of triangles is $n-2$, since each edge in an $n$ node graph is incident to at most $n-2$ triangles.

We first show that any $\eps$-edge DP algorithm for estimating the number of triangles in an arbitrary graph must have an $\Omega(n)$ additive error. Consider graphs $G_1$ and $G_2$ such that $G_1:= K_{2,n-2}$ (the complete bipartite graph where the left partite set contains 2 vertices and the right partite set contains $n-2$ vertices), and $G_2$ is $K_{2,n-2}$ with an edge between the two vertices in the left partite set. Clearly, $G_1\sim_e G_2$, and $G_1$ has zero triangles, while $G_2$ has $n-2$ triangles. A non-DP sublinear-time algorithm with multiplicative error may output 0 for $G_1$ and a value in the range $(1\pm \alpha)(n-2)$ for $G_2$, but any $\eps$-DP algorithm must output an answer with an $\Omega(n)$ additive error. We also note that when the number of triangles in the graph is $\Omega(n)$, this additive error may be absorbed into the multiplicative error.     

We apply our black-box transformation (\thmref{thm:main} and \thmref{thm:dp-transform}) to the result of~\cite{EdenLRS17} to get the following $\eps$-DP algorithm.

\begin{corollary}\corlab{cor:triangle-count}
Let $t$ denote the number of triangles in a graph. Then for any integer $K>0$ and for all $c>0,\ \eps>0$,  there exists an algorithm $\calA'$ such that 
\begin{enumerate}
    \item (Privacy) $\calA'$ is $\eps$-edge differentially private where $\eps$ is constant.
    \item (Accuracy) For all graphs $G$ on $n$ vertices, with probability $1-\frac{1}{n^c}-\frac{1}{n^{c-3}(e^\eps -1)/K+1}$,
    $$(1-\alpha')t -   \frac{2c}{\eps}\cdot {n\log(n)} - \frac{1}{Kn^3}\leq \calA'(G) \leq (1+\alpha')t + \frac{2c}{\eps}\cdot {n\log(n)}+ \frac{1}{Kn^3} $$
    where $\alpha' =\alpha \left(1+\frac{\eps}{C'\log(n)}\right)$ for some constant $C'>0$.
    \item (Query) $\calA'$ makes $\O{(\frac{n}{t^{1/3}} + \frac{m^{3/2}}{t})\poly(\log(n),\frac{1}{\alpha\eps})}$ queries. 
\end{enumerate}
\end{corollary}

\paragraph{Estimating the number of Connected Components. }
Given a graph $G$, we study the problem of estimating the number of connected components in a graph in the adjacency query list model. We recall the result by~\cite{ChazelleRT05}, whose query complexity was later improved by~\cite{BerenbrinkKM14}. 

\begin{theorem}~\cite{ChazelleRT05,BerenbrinkKM14}
Let $\mathsf{C}$ be the number of connected components in a graph with $n$ vertices. Then for $\kappa>0$, there is an algorithm that makes $\O{\frac{1}{\kappa^2}\log(\frac{1}{\kappa})\log(\frac{1}{\delta})}$ queries and with probability at least $1-\delta$, $$ \mathsf{C}-\kappa n \leq \mathsf{\tilde{C}} \leq \mathsf{C} + \kappa n \;,$$
where $\mathsf{\tilde{C}}$ denotes the output of the algorithm.
\end{theorem}

Observe that the global sensitivity of the function computing the number of connected components under edge-DP is 2, since the removal or addition of an edge can increase the number of connected components by at most 2. We apply our blackbox transformation to the result of~\cite{BerenbrinkKM14} to obtain an $\eps$-DP approximation algorithm with the following guarantees.

\begin{corollary}\corlab{cor:conn-comp}
Let the number of connected components of a graph $G$ be denoted as $\textsf{C}(G)$. Then for any integer $K>0$ and for all $c > 0, \eps > 0, \kappa > 0$, there is an algorithm $\calA'$ such that  
\begin{enumerate}
    \item (Privacy) $\calA'$ is $\eps$-differentially private where $\eps$ is constant.
    \item (Accuracy) For all graphs $G$ on $n$ vertices, with probability $1-\frac{1}{n^c} - \frac{1}{n^{c-1}(e^\eps -1)/K + 1}$,
    \begin{align*}
    &\textsf{C}(G) -\kappa n\left(\frac{8c}{\eps}+ \frac{1}{\log(n)}\right)-  \frac{2 c\log(n)}{\eps} -\frac{1}{Kn}\leq \calA'(G) \\&\leq  \textsf{C}(G)+\kappa n\left(\frac{8c}{\eps}+ \frac{1}{\log(n)}\right)+ \frac{2 c\log(n)}{\eps}   +\frac{1}{Kn}
    \end{align*}
    \item (Query) $\calA'$ makes $\O{\frac{\log^3(n)}{\kappa^2} \log\left(\frac{\log(n)}{\kappa} \right)}$ queries. 
 \end{enumerate}
\end{corollary}
 
We note that~\cite{Tetek22} give a pure DP algorithm for this problem as well, however their accuracy guarantee holds with constant probability. 
\paragraph{Estimating the weight of the MST. }
Given a (connected) graph $G=(V,E)$, with an associated weight function $wt:E(G) \to \{1,\ldots,w \}$, the problem is to estimate the weight of a minimum spanning tree in sublinear time. \cite{ChazelleRT05} reduce the problem of computing the weight of the minimum spanning tree to counting the number of connected components in various subgraphs of $G$. 

\begin{theorem}~\cite{ChazelleRT05}
Let $G$ be a graph on $n$ vertices, with an associated weight function $wt:E(G) \to \{1,\ldots,w \}$, and $w<n/2$. Let $\mathsf{M}$ be the weight of the MST of $G$.  Then for any $0<\alpha<1$, there is an algorithm that makes $\O{\bar{d}w \alpha^{-2}\log \frac{\bar{d}w}{\eps}\log(1/\delta)}$ queries and outputs a value $\mathsf{\tilde{M}}$ such that with probability at least $1-\delta$, $$ \vert \mathsf{M} - \mathsf{\tilde{M}}\vert \leq \alpha \mathsf{M} \;.$$
where $\mathsf{\tilde{M}}$ denotes the output of the algorithm. 
\end{theorem}

We first observe that the global sensitivity of the weight of the MST under edge-DP is $w$ where $w$ is the maximum weight of an edge, since the addition or deletion of an edge can change the total weight by at most $w$. By applying~\thmref{thm:main} and \thmref{thm:dp-transform} to the result of \cite{ChazelleRT05}, we obtain the following DP result for integral weights.

\begin{corollary}
Let $\mathsf{M}(G)$ be the weight of the MST of $G$. Then for all integers $K >0$ and all $c>0, \eps>0$, there exists an algorithm $\calA'$ such that 
\begin{enumerate}
    \item (Privacy) $\calA'$ is $\eps$-edge differentially private where $\eps$ is constant.
    \item (Accuracy) For all graphs $G$ on $n$ vertices, with probability $1-\frac{1}{n^c}-\frac{1}{n^{c-3}(e^\eps-1)/K + 1}$,
    $$(1-\alpha')\mathsf{M}(G) - \frac{2cw \log(n)}{\eps} - \frac{1}{Kn^3} \leq \calA'(G) \leq (1+\alpha')\mathsf{M}(G) + \frac{2cw \log(n)}{\eps}+\frac{1}{Kn^3}$$
    where $\alpha' =\alpha \left(1+\frac{\eps}{C'\log(n)}\right)$ for some constant $C'>0$.
    \item (Query) $\calA'$ makes $\O{\bar{d}w\frac{\log^3(n)}{\alpha^2\eps^2}\log\left(\frac{\bar{d}w\log(n)}{\alpha\eps}\right)}$ queries.
\end{enumerate}
\end{corollary}

\paragraph{Estimating the average degree.}Recall that $\bar{d}$ denotes the average degree of a graph. We study the problem of estimating the average degree of the graph in the adjacency query list model. We first recall the result of~\cite{GoldreichR04} below, 
\begin{theorem}~\cite{GoldreichR04}
For every $\alpha \in (0,1)$, there is an algorithm that makes $\O{\frac{n}{\sqrt{m}} \poly\left(\frac{\log(n)}{\alpha}\right)\log(1/\delta)}$ queries and outputs $\tilde{d}$ such that with probability at least $1-\delta$,
$$\vert \bar{d} - \tilde{d} \vert \leq \alpha\bar{d} 
\;.$$
where $\tilde{d}$ denotes the output of the algorithm. 
\end{theorem}
We will first present our result using our DP framework and then compare this result to existing DP results. In particular, by applying our blackbox transformation (\thmref{thm:main} and \thmref{thm:dp-transform}) directly to the result of~\cite{GoldreichR04}, and observing that the global sensitivity of the average degree under edge-DP is $2/n$ (since the addition or deletion of an edge affects two vertices in the sum), we obtain the following.  
\begin{corollary}\corlab{cor:avg-deg}
Let $\bar{d}$ denote the average degree of a graph. Then for any integer $K>0$ and for all $c>0,\ \eps>0$,  there exists an algorithm $\calA'$ such that 
\begin{enumerate}
    \item (Privacy) $\calA'$ is $\eps$-edge differentially private where $\eps$ is constant.
    \item (Accuracy) For all graphs $G$ on $n$ vertices, with probability $1-\frac{1}{n^c}-\frac{1}{n^{c-1}(e^\eps-1)/K + 1}$,
    $$(1-\alpha')\bar{d} -   \frac{4c}{\eps}\cdot \frac{\log(n)}{n} - \frac{1}{Kn} \leq \calA'(G) \leq (1+\alpha')\bar{d} +   \frac{4c}{\eps}\cdot \frac{\log(n)}{n} + \frac{1}{Kn} $$
    where $\alpha' =\alpha \left(1+\frac{\eps}{C'\log(n)}\right)$ for some constant $C'>0$.
    \item (Query) $\calA'$ makes $\O{\frac{n}{\sqrt{m}} \poly\left(\frac{\log^2(n)}{\eps\alpha}\right)\log(n)}$ queries.
\end{enumerate}
\end{corollary}
In recent work,~\cite{BlockiGM22} adapted the algorithm by~\cite{GoldreichR04} to obtain an $\eps$-DP $(\alpha,0,\delta)$-approximation of $\bar{d}$ in $\O{\sqrt{n} \poly\left(\frac{\log(n)}{\alpha}\right)\poly\left(\frac{1}{\eps}\right) \log(1/\delta)}$ for $\bar{d} \geq 1$ queries. 
First observe that,~\cite{BlockiGM22} achieves a multiplicative approximation, whereas~\corref{cor:avg-deg} has an additive error (along with the multiplicative error). Thus, it may seem that the algorithm given by~\cite{BlockiGM22} is a better algorithm overall. But the algorithm in~\corref{cor:avg-deg} has a few advantages. The first advantage is that unlike~\cite{BlockiGM22}, the algorithm works for all values $\bar{d} \geq 0$, whereas~\cite{BlockiGM22} only works for the (albeit natural) assumption that $\bar{d}\geq 1$. Also for reasonable values of $\bar{d}$, i.e., for $\bar{d} \geq \Omega(\frac{\log n}{n})$, the additive error above is absorbed into the multiplicative error thus resulting in a pure DP algorithm with multiplicative error and no prior assumptions on the average degree $\bar{d}$. The query complexity achieved in~\corref{cor:avg-deg} can also be much better than ~\cite{BlockiGM22}. For example, if $m = \Omega(n^2)$ then our algorithm runs in $\poly \left(\frac{\log(n)}{\alpha\eps} \right)$ queries. In general if the graph has $m = \omega(n)$ edges then the query complexity of the original non-DP algorithm~\cite{GoldreichR04} is  superior to \cite{BlockiGM22}, and the query complexity of our algorithm is comparable to ~\cite{GoldreichR04}  ---  only incurs a multiplicative factor of $\poly \left(\frac{\log(n)}{\alpha\eps} \right)$ overhead.  Lastly, due to its black-box nature, the algorithm in~\corref{cor:avg-deg} is much simpler in nature. The algorithm of~\cite{BlockiGM22} is more complex due to its careful addition of noise in a way that results in (essentially) the same approximation guarantees as the non-DP algorithm. 

We note that \cite{GoldreichR04} avoided assuming any lowerbound on the average degree by first designing a sublinear-time algorithm that estimates the average degree by taking a degree lower bound $\ell$ as an input parameter and then removing this lower bound parameter via a geometric search. Implementing this type of geometric-search procedure in the DP setting would result in a privacy loss of $\log(n)$, and thus~\cite{BlockiGM22} avoided this step by apriori assuming $\bar{d} \geq 1$. In recent work~\cite{Tetek22} gives a general technique of making this procedure DP and thus achieves a pure DP algorithm for the average degree problem with (essentially) the same guarantees as the non-DP algorithm and no prior assumption on $\bar{d}$. However, their accuracy guarantees only hold with constant probability, whereas the accuracy guarantees achieved by applying our DP framework hold with high probability.  

%

\paragraph{Estimating the maximum matching size. }
We study the problem of estimating the maximum matching size in a graph of maximum degree $d$ in the adjacency query list model. \cite{YoshidaYI12} gave an $(0,\kappa n)$-approximation for this problem with an expected query complexity of $d^{\O{1/\kappa^2}}$. 
In order to achieve a high probability bound on the query-complexity, we run $\Theta(\log(1/\delta))$ instances of the algorithm in parallel and return the output of the instance that terminates first.

\begin{theorem}~\cite{YoshidaYI12}
For all $\kappa>0$, there is a $(0,\kappa n,\delta)$-approximation algorithm for the maximum matching problem that uses $\O{d^{\O{1/\kappa^2}}\log(1/\delta)}$ queries with probability at least $1-\delta$. 
\end{theorem}

We observe that the global sensitivity of the size of a maximum matching of a graph under node (and edge-DP) is at most one, since the node that is added/removed could have only contributed to at most one edge in the matching. By applying our black-box DP transformation to the results of~\cite{YoshidaYI12}, we get the following guarantees. 
\begin{corollary}\corlab{cor:max-match}
Let the size of the maximum matching of a graph $G$ be denoted as $\mathsf{MM}(G)$. Then for any integer $K>0$ and for all $c > 0,\ \eps > 0$, there is an algorithm $\calA'$ such that  
\begin{enumerate}
    \item (Privacy) $\calA'$ is $\eps$-node (and edge) differentially private where $\eps$ is constant.
    \item (Accuracy) For all graphs $G$ on $n$ vertices, with probability $1-(\frac{1}{n^c}+\frac{1}{n^{c-1}(e^\eps-1)/K + 1}+\frac{1}{10})$,
     \begin{align*}
    &\mathsf{MM}(G) -\kappa n\left(\frac{8 \log 10}{\eps}+ {1}\right)-  \frac{2 \log 10}{\eps} - \frac{1}{Kn}\leq \calA'(G) \\&\leq  \mathsf{MM}(G)+\kappa n\left(\frac{8 \log 10}{\eps}+ {1}\right)+ \frac{2 \log 10}{\eps} + \frac{1}{Kn}    
    \end{align*} 
    \item (Query) $\calA'$ makes $\O{d^{\O{1/\kappa^2}}\log(n)}$ queries. 
 \end{enumerate}
\end{corollary}

In recent work, \cite{BlockiGM22} obtains an approximation for the maximum matching size with a multiplicative factor of 2 and additive factor of $\kappa n$ in time $\tilde{O}\left((\bar{d}+1)/\kappa^2 \right)$. They left the problem of designing a DP algorithm that achieves a $(0, \kappa n, \delta)$-approximation as an open problem. \corref{cor:max-match} partially resolves this open problem. We note that~\cite{Tetek22} fully resolves this problem by giving a pure DP algorithm whose accuracy guarantee holds with high probability. 

One disadvantage of our framework is that if the resource (e.g., time, query or space) complexity of the non-DP approximation algorithm depends on the approximation parameter in a non-polynomial way (e.g., exponentially), then the resulting DP algorithm that achieves an accuracy of similar magnitude with high probability may be inefficient. This is the case for the algorithm of~\cite{YoshidaYI12} --- the non-DP approximation algorithm has a query complexity of $\O{d^{\O{1/\kappa^2}}\log(1/\delta)}$ (which is exponential in the approximation parameter $\kappa$). If we wanted to guarantee that the additive error due to sampling from the Laplace distribution is small \emph{with high probability} (e.g., with probability $\geq 1-\frac{1}{n}$), then we would need to set $\gamma=\log(n)$ in \thmref{thm:main}. Without tuning $\kappa$, the resulting accuracy would have an error of $\O{n\log n}$, which would render the output of the algorithm meaningless (since the matching size $\leq n$). On the otherhand, if we tune $\kappa$ to $\tau:=\kappa/\log(n)$ in \algref{alg:fpras-DP}, the approximation achieved would be meaningful, i.e., the additive error would be small enough, but the query complexity of the resulting DP algorithm would become $\O{d^{\O{\log^2(n)/\kappa^2}}\log(1/\delta)}$. Thus in order to achieve the guarantees in \corref{cor:max-match}, we need to relax the guarantee that the additive error due to sampling from the Laplace distribution is small \emph{with high probability}, and ensure this is true \emph{with constant probability} instead (for e.g., with probability $\geq 9/10$).  

\paragraph{Estimating the distance to bipartiteness. } We study the problem of estimating the distance to bipartiteness in the adjacency list query model. The distance to the graph property of bipartiteness is said to be $\psi$ if $\psi n^2$ edges need to be added or removed from the graph on $n$ vertices in order to obtain a bipartite graph. We first recall the result by~\cite{GhoshMRS22} as follows. 

\begin{theorem}~\cite{GhoshMRS22}
For every $\kappa > 0$, there exists an algorithm that given input graph $G$ inspects a random subgraph of $G$ on $\tO{(1/\kappa^3)\log(1/\delta)}$ vertices and estimates the distance from $G$ to bipartiteness to within an additive error of $\kappa n^2$ with probability $\geq 1-\delta$. 
\end{theorem}
We observe that the global sensitivity (under edge-DP) of the distance to bipartiteness function is $1/n^2$ (or 1 without the normalization factor), since by adding an edge one can introduce an odd cycle, making a previously bipartite graph into a non-bipartite graph, and conversely, by removing an edge, one can remove the existence of an odd cycle thus making a previously non-bipartite graph into a bipartite graph.  By applying \thmref{thm:main} to the results of~\cite{GhoshMRS22}, we get the following guarantees.

\begin{corollary}\corlab{cor:dist-bipart}
Let the distance from a graph $G$ to bipartiteness be denoted as $\mathsf{d_B}(G)$. For any integer $K>0$ and for all $c > 0, \eps > 0$, there is an algorithm $\calA'$ such that  
\begin{enumerate}
    \item (Privacy) $\calA'$ is $\eps$-differentially private where $\eps$ is constant.
    \item (Accuracy) For all graphs $G$ on $n$ vertices, with probability $1-\frac{1}{n^c}-\frac{1}{n^{c-2}(e^\eps-1)/K+1}$,
    \begin{align*}\small
    &\mathsf{d_B}(G) -\kappa n^2\left(\frac{8c}{\eps}+ \frac{1}{\log(n)}\right)-  \frac{2c \log(n)}{\eps}-\frac{1}{Kn^2} \leq \calA'(G) \\&\leq  \mathsf{d_B}(G)+\kappa n^2\left(\frac{8c}{\eps}+ \frac{1}{\log(n)}\right)+ \frac{2c\log(n)}{\eps} +\frac{1}{Kn^2}
    \end{align*}
    \item (Query) $\calA'$ inspects $\tO{\frac{\log^4(n)}{\kappa^3}}$ vertices.
 \end{enumerate}
\end{corollary}

\section{Private Streaming Algorithms}\seclab{sec:streaming}

In this section, we apply our general framework to achieve private algorithms on data streams, where updates to an underlying dataset arrive sequentially and the goal is to output or approximate some predetermined function using space sublinear in the size of the input. 
In \secref{sec:dynamic}, we consider applications to dynamic/turnstile streams, where updates of the stream can be both insertions and deletions.
In \secref{sec:sw}, we consider applications to the sliding window model, where all updates of the stream are insertions, but elements are also implicitly deleted by the sliding window once they are past their expiration, i.e., when $W$ subsequent updates arrive in the stream, for a fixed window parameter $W>0$.

\subsection{Dynamic/Turnstile Streams}
\seclab{sec:dynamic}
In this section, we first apply our general framework to problems in the dynamic/turnstile streaming model. 
We remark that since this model generalizes the insertion-only model, our results also automatically apply to the insertion-only model. 
We say that streams $\frakS$ and $\frakS'$ are \emph{neighboring} if there exists a single update $i\in[m]$ such that $u_i\neq u'_i$, where $u_1,\ldots,u_m$ are the updates of $\frakS$ and $u'_1,\ldots,u'_m$ are the updates of $\frakS'$. The DP streaming algorithms we present here are in the one-shot model, i.e., the algorithm outputs at the end of the stream. 

\paragraph{Weighted minimum spanning tree.}
We first study the problem of estimating the weighted minimum spanning tree of a graph implicitly defined in a data stream. 
Given a set of vertices $V$, each update in the stream has the form $[-M,M]\times(u,v)$ for some $u,v\in V$, which changes the weight of edge $(u,v)$ by some amount between $-M$ and $M$. 
We recall the following turnstile streaming algorithm for estimating the weighted minimum spanning tree by~\cite{AhnGM12}. 
\begin{theorem}~\cite{AhnGM12}
\thmlab{thm:mst:turnstile}
For any $\alpha\in(0,1]$, there exists a single-pass algorithm that outputs a $(\alpha,0)$-approximation to the weight of the minimum spanning tree of a dynamic graph with probability at least $\frac{2}{3}$, using $\O{\frac{1}{\alpha}n\log^3 n}$ bits of space. 
\end{theorem}

\begin{lemma}
\lemlab{lem:mst:sens}
Suppose a graph is implicitly defined by a dynamic stream that changes the weight of each edge between $[-M,M]$. 
Then the global sensitivity of the weighted minimum spanning tree is at most $2M$. 
\end{lemma}
\begin{proof}
Let $G,G'$ be two connected graphs so that some edge $(u,v)$ has weight $W$ in $G$ and $W'$ in $G'$ and suppose without loss of generality that $W<W'$. 
Let $T$ be the minimum weighted spanning tree of $G$ and $T'$ be the minimum weighted spanning tree of $G'$. 
Note that since $W'>W$, then $T$ is also a valid tree in $G$, so that the minimum weighted spanning tree of $G'$ has at most the weight of $T$ in $G'$. 
Since the weight of $(u,v)$ changed by at most $M$, then the weight of $T$ in $G'$ is at most $M$ more than the weight of $T$ in $G$. 
Hence, the minimum weighted spanning tree of $G'$ has weight at most $M$ more than that of the minimum weighted spanning tree of $G$. 
Moreover, the weight of each edge in $G'$ is at least the weight of each edge in $G$, so the minimum weighted spanning tree of $G'$ has weight at least that of the minimum weighted spanning tree of $G$. 
Thus the minimum weighted spanning trees of $G$ and $G'$ differ by at most $M$. 
By the triangle inequality, it follows that the global sensitivity of the weighted minimum spanning tree of a dynamic graph is at most $2M$. 
\end{proof}

Therefore, by first applying the median trick (see~\remref{rem:median-trick}) to \thmref{thm:mst:turnstile} to boost the success probability and then applying our main framework, and noting the global sensitivity bounds in \lemref{lem:mst:sens}, we have the following guarantee:
\begin{corollary}
\corlab{cor:mst:turnstile}
For any integer $K>0$ and for all $\alpha\in(0,1)$, $c>0$, $\eps>0$, and $\rho:= \O{\frac{\alpha\eta}{\log m}}$ on a turnstile stream $\frakS$ of length $m=\poly(n)$, where $\eta=\min(\eps,1)$, there exists a turnstile streaming algorithm $\calA$ algorithm that:
\begin{enumerate}
\item (Privacy) The algorithm $\calA$ is $\eps$-differentially private where $\eps$ is constant.
\item (Accuracy) The algorithm $\calA$ outputs $\widehat{W}(\frakS)$ such that with probability at least $1-\frac{1}{m^c} - \frac{1}{m^{c-1}(e^\eps-1)/(KMm) + 1}$, 
\[(1-\alpha)W(\frakS)-\frac{12cM\log m}{\eps} - \frac{1}{KMm}\le\widehat{W}(\frakS)\le(1+\alpha)W(\frakS)+\frac{12cM\log m}{\eps} +  \frac{1}{KMm},\]
where $W(\frakS)$ denotes the minimum weighted spanning tree in the graph induced by the stream $\frakS$. 
\item (Space) The algorithm $\calA$ uses space $\O{\frac{1}{\alpha\eta} n \log^5 n}$ bits. 
\end{enumerate}
\end{corollary}

\paragraph{$L_p$ norm and $F_p$ moment estimation.}
We next study the related problems of $L_p$ norm and $F_p$ moment estimation, where a frequency vector $x\in\mathbb{R}^n$ is implicitly defined by updates of the stream and the goal is to approximate the the $F_p$ moment of $x$, i.e., $x_1^p+\ldots+x_n^p$, or the $L_p$ norm of $x$, i.e., $\left(x_1^p+\ldots+x_n^p\right)^{1/p}$. 
Each update of the stream has the form $\{-1,+1\}\times[n]$, indicating whether the update implicitly decrements or increments the corresponding coordinate of the frequency vector $x$. 
Finally, the $L_0$ norm of $x$ is defined as $\|x\|_0=|\{i\in[n]\,|\,x_i\neq 0\}|$, or the number of nonzero coordinates of $x$. 
We assume the data stream has length $m$.

We first upper bound the global sensitivity of the $L_p$ norm for $p\ge 1$ and the $F_p$ moment for $p\in[0,1]$. 
\begin{lemma}
\lemlab{lem:lp:sens}
For $p\ge 1$, the global sensitivity of the $L_p$ norm is at most $2$.
For $p\in[0,1]$, the global sensitivity of the $F_p$ moment is at most $2$. 
\end{lemma}
\begin{proof}
Let $x,x'\in\mathbb{R}^n$ be two vectors defined by neighboring streams, so that $\|x-x'\|_1\le 2$. 
Then for $p\ge 1$, we have $\|x-x'\|_p\le\|x-x'\|_1\le 2$. 
For $p\in(0,1]$, it suffices to note that $|y^p-(y-1)^p|\le 1$ and at most two coordinates differ between $x$ and $x'$, each by at most one. 
Finally, for $p=0$, note that since at most two coordinates differ between $x$ and $x'$, then $|x-x'|_0\le 2$. 
\end{proof}

For $p>2$, we use the following turnstile streaming algorithm of \cite{GangulyW18}: 
\begin{theorem}~\cite{GangulyW18}
\thmlab{thm:lp:turnstile:big}
For $p>2$, and accuracy parameter $\alpha>0$ there is a turnstile streaming algorithm for $(\alpha,0)$-approximation of $L_p$ with probability at least $1-\delta$ that uses \sloppy $\O{\frac{1}{\alpha^2} n^{1-\frac{2}{p}}\log\frac{1}{\delta}\log n + \frac{1}{\alpha^{4/p}} n^{1-\frac{2}{p}} \log^{2/p}\frac{1}{\delta} \log^2 n}$ bits of space. 
\end{theorem}
Thus by applying our main framework from \thmref{thm:main} and \thmref{thm:dp-transform} to \thmref{thm:lp:turnstile:big} and noting the global sensitivity bounds in \lemref{lem:lp:sens}, we have the following guarantee:
\begin{corollary}
\corlab{cor:lp:turnstile:big}
For any integer $K>0$, and for all $p>2$, $\alpha>0$, $c>0$, $\eps>0$, and $\rho:= \O{\frac{\alpha\eta}{\log m}}$ on a stream $\frakS$ of length $m=\poly(n)$, where $\eta=\min(\eps,1)$, there exists a turnstile streaming algorithm $\calA$ algorithm that:
\begin{enumerate}
\item (Privacy) The algorithm $\calA$ is $\eps$-differentially private where $\eps$ is constant.
\item (Accuracy) The algorithm $\calA$ outputs $\widehat{L_p}(\frakS)$ such that with probability at least $1-\frac{1}{m^c}-\frac{1}{\poly(m)(e^\eps -1)/K +1}$, 
\[(1-\alpha)L_p(\frakS)-\frac{12c\log m}{\eps} - \frac{1}{K \poly(m)}\le\widehat{L_p}(\frakS)\le(1+\alpha)L_p(\frakS)+\frac{12c\log m}{\eps} + \frac{1}{K \poly(m)},\]
where $L_p(\frakS)$ denotes the $L_p$ norm of the frequency vector induced by the stream $\frakS$. 
\item (Space) The algorithm $\calA$ uses space $\O{\frac{p}{ \alpha^2\eta^2} n^{1-2/p}}\cdot\poly\left(\log n,\log\frac{1}{\alpha\eta}\right)$ bits.
\end{enumerate}
\end{corollary}

When $p=2$, we use the following well-known AMS algorithm~\cite{AlonMS99}:
\begin{theorem}~\cite{AlonMS99}
\thmlab{thm:lp:turnstile:two}
For an accuracy parameter $\alpha>0$ there is a turnstile streaming algorithm for $(\alpha,0)$-approximation of $L_2$ with probability at least $1-\delta$ that uses $\O{\frac{1}{\alpha^2}\log m\log\frac{1}{\delta}}$ bits of space, for a stream of length $m$.  
\end{theorem}
We note that~\cite{Tetek22} give an $\eps$-DP algorithm using the AMS algorithm as well, but the accuracy guarantee only holds with constant probability. By applying our main framework from \thmref{thm:main} and \thmref{thm:dp-transform} to the AMS algorithm in \thmref{thm:lp:turnstile:two} and noting the global sensitivity bounds in \lemref{lem:lp:sens}, we have the following guarantee:
\begin{corollary}
\corlab{cor:lp:turnstile:two}
For any integer $K>0$ and for any $\alpha>0$, $c>0$, $\eps>0$, and $\rho:= \O{\frac{\alpha\eta}{\log m}}$ on a stream $\frakS$ of length $m=\poly(n)$, where $\eta=\min(\eps,1)$, there exists a turnstile streaming algorithm $\calA$ algorithm that:
\begin{enumerate}
\item (Privacy) The algorithm $\calA$ is $\eps$-differentially private where $\eps$ is constant.
\item (Accuracy) The algorithm $\calA$ outputs $\widehat{L_2}(\frakS)$ such that with probability at least $1-\frac{1}{m^c}-\frac{1}{\poly(m)(e^\eps -1)/K +1}$, 
\[(1-\alpha)L_2(\frakS)-\frac{12c\log m}{\eps} - \frac{1}{K\poly(m)}\le\widehat{L_2}(\frakS)\le(1+\alpha)L_2(\frakS)+\frac{12c\log m}{\eps}+\frac{1}{K \poly(m)},\] 
where $L_2(\frakS)$ is the $L_2$-norm of the stream $\frakS$. 
\item (Space) The algorithm $\calA$ uses space $\O{\frac{1}{\alpha^2\eta^2}\log^4 n}$ bits.
\end{enumerate}
\end{corollary}

For $p\in(0,2)$, we use the following turnstile streaming algorithm of~\cite{KaneNPW11}:
\begin{theorem}~\cite{KaneNPW11}
\thmlab{thm:lp:turnstile:med}
For $p\in(0,2)$, and accuracy parameter $\alpha>0$ there is a turnstile streaming algorithm for $(\alpha,0)$-approximation of $L_p$ with probability at least $\frac{2}{3}$ that uses $\O{\frac{1}{\alpha^2}\log m+\log\log n}$ bits of space. 
\end{theorem}
Hence by applying the median trick to \thmref{thm:lp:turnstile:med} to boost the success probability, and then applying our main framework from \thmref{thm:main}, and finally noting the global sensitivity bounds in \lemref{lem:lp:sens}, we have the following guarantee:
\begin{corollary}
\corlab{cor:lp:turnstile:small}
For any integer $K>0$, $p\in(0,2)$, and for any $\alpha>0$, $c>0$, $\eps>0$, and $\rho:= \O{\frac{\alpha\eta}{\log m}}$ on a stream $\frakS$ of length $m=\poly(n)$, where $\eta=\min(\eps,1)$, there exists a turnstile streaming algorithm $\calA$ algorithm that:
\begin{enumerate}
\item (Privacy) The algorithm $\calA$ is $\eps$-differentially private where $\eps$ is constant.
\item (Accuracy) The algorithm $\calA$ outputs $\widehat{X}$ such that with probability at least $1-\frac{1}{m^c}-\frac{1}{\poly(m)(e^\eps -1)/K +1}$, 
\[(1-\alpha)X-\frac{12c\log m}{\eps} - \frac{1}{K \poly(m)}\le\widehat{X}\le(1+\alpha)X+\frac{12c\log m}{\eps}+\frac{1}{K \poly(m)},\]
where $X$ is the $L_p$ norm of the stream $\frakS$ for $p\in[1,2)$ and $X$ is the $F_p$ moment of the stream $\frakS$ for $p\in(0,1)$. 
\item (Space) The algorithm $\calA$ uses space $\O{\frac{1}{\alpha^2\eta^2}\log^4 n}$ bits. 
\end{enumerate}
\end{corollary}
For $p=0$, we use the following turnstile streaming algorithm of~\cite{KaneNW10}:
\begin{theorem}~\cite{KaneNW10}
\thmlab{thm:distinct:turnstile}
For any accuracy parameter $\alpha>0$ there is a turnstile streaming algorithm for $(\alpha,0)$-approximation of $L_0$ with probability at least $\frac{2}{3}$ that uses $\O{\frac{1}{\alpha^2}\log n\log\frac{1}{\alpha}+\log\log m}$ bits of space. 
\end{theorem}
Then by applying the median trick to \thmref{thm:lp:turnstile:med} to boost the probability of success to $1-\delta$, and then applying our main framework from \thmref{thm:main}, and finally noting the global sensitivity bounds in \lemref{lem:lp:sens}, we have the following guarantee:
\begin{corollary}
\corlab{cor:distinct:turnstile}
For any integer $K>0$, and for any $\alpha>0$, $c>0$, $\eps>0$, and $\rho:= \O{\frac{\alpha\eta}{\log m}}$ on a stream $\frakS$ of length $m=\poly(n)$, where $\eta=\min(\eps,1)$, there exists a turnstile streaming algorithm $\calA$ algorithm that:
\begin{enumerate}
\item (Privacy) The algorithm $\calA$ is $\eps$-differentially private where $\eps$ is constant.
\item (Accuracy) The algorithm $\calA$ outputs $\widehat{L_0}(\frakS)$ such that with probability at least $1-\frac{1}{m^c}-\frac{1}{\poly(m)(e^\eps -1)/K +1}$, 
\[(1-\alpha)L_0(\frakS)-\frac{12c\log m}{\eps}-\frac{1}{K \poly(m)}\le\widehat{L_0}(\frakS)\le(1+\alpha)L_0(\frakS)+\frac{12c\log m}{\eps}+\frac{1}{K \poly(m)},\]
where $L_0(\frakS)$ is the number of distinct elements in the stream $\frakS$. 
\item (Space) The algorithm $\calA$ uses space $\tO{\frac{1}{\alpha^2\eta^2}\log^4 n}$ bits.
\end{enumerate}
\end{corollary}

\subsection{Sliding Windows}
\seclab{sec:sw}
In this section, we introduce the sliding window model and apply our DP framework to ensure privacy for many sliding window algorithms. 

\begin{definition}[Sliding window model]
In the sliding window model, there exists a data stream of length $m$ and a window parameter $W>0$. 
The underlying dataset is then implicitly defined by only the $W$ most recent updates in the stream. 
\end{definition}

The sliding window model is a generalization of the insertion-only streaming model that captures the prioritization of recent data over outdated or expired data and thus there are a number of time sensitive applications applications~\cite{BabcockBDMW02,DatarM07,MankuM12,PapapetrouGD15,EpastoLVZ17}, such as network monitoring~\cite{CormodeM05a,CormodeG08} or event detection on social media~\cite{OsborneEtAl2014}, for which the sliding window model has better performance than the insertion-only streaming model. 
The sliding window model is especially appropriate for time-sensitive applications such as network monitoring~\cite{CormodeM05a,CormodeG08} and has been subsequently studied in a number of additional settings~\cite{DatarM07,BravermanOZ12,EpastoLVZ17,BravermanGLWZ18, BravermanDMMUWZ20,BravermanWZ21,EpastoMMZ22,JayaramWZ22}.

For the purposes of differential privacy, we permit two neighboring streams to differ by a single update over the entire data stream. 
Observe that this notion includes the setting where two neighboring streams differ by a single update in the dataset induced by the sliding window, i.e., the last $W$ updates of the data stream. 

\cite{BravermanO10} presented the smooth histogram framework~\cite{BravermanO10}, which converts an existing insertion-only streaming algorithm into a sliding window algorithm. 
We remark that unfortunately, due to standard nomenclature in previously non-intersecting literature, the notion of smooth histogram does not use the same notion of smooth as smooth sensitivity. 
Hence, we first define the following notion of a smooth function:

\begin{definition} \deflab{def:smooth}
A function $f \ge 1$ is $(\rho,\xi)$-smooth if it has the following properties:
\begin{description}
\item [Monotonicity]
$f(A)\ge f(B)$ for $B\subseteq A$ ($B$ is a suffix of $A$)
\item [Polynomial boundedness]
There exists $c>0$ such that $f(A)\le n^c$.
\item [Smoothness]
For any $\rho\in(0,1)$, there exists $\xi\in(0,\rho]$ so that if $B\subseteq A$ and $(1-\xi)f(A)\le f(B)$, then $(1-\rho)f(A\cup C)\le f(B\cup C)$ for any adjacent $C$.
\end{description}
\end{definition}
The smooth histogram framework has the following guarantees:
\begin{theorem}[Smooth Histogram Framework,~\cite{BravermanO10}]
\thmlab{thm:smooth:histogram}
Let $f:\calU^m \to \mathbb{R}$ be a $(\rho, \beta(\rho))$-smooth function for any $\rho\in(0,1)$ on a stream $\frakS$ of length $m=\poly(n)$, and suppose there exists an insertion-only streaming algorithm $\calA$ that outputs a $(\alpha,0)$-approximation of $f$ using space $\calS(\alpha,\delta,m,n)$ and update time $\calT(\alpha,\delta,m,n)$, for any approximation parameter $\alpha\in(0,1)$ and failure probability $\delta$. 
Then for any constant $c>0$ and window parameter $W>0$, there exists a sliding window algorithm $\calA'$ such that for the dataset $\calW$ induced by the last $\min(m,W)$ updates of the stream:
\begin{enumerate}
\item (Accuracy) The algorithm $\calA'$ outputs $\widehat{f}(\calW)$ such that with probability at least $1-\frac{1}{m^c}$, 
\[(1-\alpha)f(\calW)\le\widehat{f}(\calW)\le(1+\alpha)f(\calW).\]
\item (Time/Space) The algorithm $\calA'$ uses space $\O{\frac{1}{\beta(\alpha) }(\calS(\beta(\alpha),\frac{\delta}{\poly(m,n)},m,n)+\log m)\log m}$ and update time $\O{\frac{1}{\beta(\alpha)}(\calT(\beta(\alpha),\frac{\delta}{\poly(m,n)},m,n))\log m}$.
\end{enumerate}
\end{theorem}

The main takeaway for the smooth histogram framework is that it essentially provides a means to achieve a $(\rho,0)$-approximation algorithm in the sliding window model for a smooth function for which there already exists a $(\rho,0)$-approximation algorithm in the insertion-only streaming model. 
For many additional problems, either there are no known $(\rho,0)$-approximation algorithms for all $\rho>0$, e.g.,~\cite{EpastoLVZ17} or the output is not a scalar quantity, e.g.,~\cite{BravermanOZ12,BravermanDMMUWZ20,BravermanWZ21,EpastoMMZ22,JayaramWZ22}. 

By \thmref{thm:smooth:histogram} and \thmref{thm:main}, we have the following guarantee:
\begin{theorem}[DP Smooth Histogram Framework]
\thmlab{thm:main-hist-framework}
Let $\alpha >0, c > 0, \eps > 0$ and a $(\rho, \beta(\rho))$-smooth function $f: \calU^m \to \mathbb{R}$ for $\rho:= \O{\frac{\alpha\eta}{\log m}}$ on a stream $\frakS$ of length $m=\poly(n)$, where $\eta=\min(\eps,1)$ and suppose there exists an insertion-only streaming algorithm $\calA$ that outputs a $(\alpha,0)$-approximation of $f$ using space $\calS(\alpha,\delta,m,n)$ and update time $\calT(\alpha,\delta,m,n)$, and failure probability $\delta$. 
Then there exists a sliding window algorithm $\calA'$ such that for the dataset $\calW$ induced by the last $\min(m,W)$ updates of the stream:
\begin{enumerate}
\item (Privacy) The algorithm $\calA'$ is $(\eps,\delta)$-differentially private where $\eps$ is constant and $\delta=\frac{1}{m^c}$.  
\item (Accuracy) The algorithm $\calA'$ outputs $\widehat{f}(\calW)$ such that with probability at least $1-\frac{1}{m^c}$, 
\[(1-\alpha)f(\calW)-\frac{6c\Delta_f\log m}{\eps}\le\widehat{f}(\calW)\le(1+\alpha)f(\calW)+\frac{6c\Delta_f\log m}{\eps}.\]
\item (Time/Space) The algorithm $\calA'$ uses space $\O{ \frac{1}{\beta(\rho) }(\calS(\beta(\rho),\frac{\delta}{\poly(m,n)},m,n)+\log m)\log m}$ and update time $\O{\frac{1}{\beta(\rho)}(\calT(\beta(\rho),\frac{\delta}{\poly(m,n)},m,n))\log m}$.
\end{enumerate}
\end{theorem}

Observe that if we tried to apply the non-private smooth histogram framework to a DP insertion-only streaming algorithm, this might preserve privacy (by post-processing), but may significantly increase the error (in terms of accuracy). 

\paragraph{Length of the longest increasing sub-sequence.} 
The longest increasing subsequence problem $LIS$ is to find an increasing subsequence of maximum length of a sequence whose elements are from the universe $[n]$ and sequentially defined by the stream.  
The $LIS$-$length$ problem is to output the length of such a subsequence. 
\cite{BravermanO07} showed that the $LIS$-$length$ function is $(\rho,\rho)$-smooth. 
\begin{theorem}~\cite{BravermanO07}
For any $\rho\in(0,1)$, LIS-length is a $(\rho,\rho)$-smooth function. 
\end{theorem}
We now recall the following insertion-only streaming algorithm for estimating the length of the longest increasing subsequence

\begin{theorem}~\cite{SunW07}
Let $\frakS$ be a stream $u_1, u_2, \ldots $ such that $u_i \in [n]$ and let $k$ be an upper bound on the LIS-length in the stream, then there exists a 1-pass streaming algorithm that computes LIS-length with probability at least $\frac{2}{3}$, using space $\O{k^2\log\frac{n}{k}}$. 
\end{theorem}

We can therefore apply \thmref{thm:main-hist-framework} and \thmref{thm:dp-transform} to the algorithm of \cite{SunW07} to obtain \corref{cor:lis}. 
We note that while the longest increasing subsequence may have large global sensitivity, the length of the longest increasing subsequence only has global sensitivity 1, i.e., $\Delta_{LIS\text{-}length}=1$. 
\begin{lemma}
\lemlab{lem:lis:sens}
For the length of the longest increasing subsequence, the global sensitivity is at most $1$. 
\end{lemma}
\begin{proof}
Let $S,S'$ be two sequences so that $S'$ is formed by deleting a term of $S$. 
Since any subsequence of $S'$ is a subsequence of $S$ and any subsequence of $S$ can be transformed into a subsequence of $S'$ by deleting at most one term, then the lengths of the longest subsequences in $S$ and $S'$ differ by at most $1$. 
Hence by the triangle inequality, the global sensitivity of the length of the longest increasing subsequence is at most $2$. 
\end{proof}

\begin{corollary}\corlab{cor:lis}
\sloppy Let $\frakS$ be a stream $u_1, u_2, \ldots $ such that $u_i \in [n]$ of length $m=\poly(n)$ and let $k$ be an upper bound on the LIS-length in the stream. 
For some $K>0$, and for any $\alpha >0$, $c>0$, $\eps>0$, and $\rho:= \O{\frac{\alpha\eta}{\log m}}$ on a stream of length $m$, where $\eta=\min(\eps,1)$, there exists a sliding window algorithm $\calA$ such that:
\begin{enumerate}
\item (Privacy) The algorithm $\calA$ is $\eps$-differentially private where $\eps$ is constant.
\item (Accuracy) The algorithm $\calA$ outputs $\widehat{LIS}(\calW)$ such that with probability at least $1-\frac{1}{m^c}-\frac{1}{m^c(e^\eps-1)/(Kk) + 1}$, 
\[(1-\alpha)LIS(\calW)-\frac{12c\log m}{\eps}-\frac{1}{Kk}\le\widehat{LIS}(\calW)\le(1+\alpha)LIS(\calW)+\frac{12c\log m}{\eps}+\frac{1}{Kk},\]
where $LIS(\calW)$ denotes the length of the longest increasing subsequence of the last $\min(m,W)$ updates of $\frakS$, for a window parameter $W>0$. 
\item (Time/Space) \sloppy The algorithm $\calA$ uses an $\O{\frac{k^2}{\alpha\eta} \log^4 n}$ space.
\end{enumerate}

\end{corollary}

\paragraph{Distinct elements.} 
We first give a simple proof to show that the function computing the number of distinct elements is $(\rho,\rho)$-smooth below. 
\begin{lemma}
For any $\rho\in(0,1)$, the number of distinct elements is $(\rho,\rho)$-smooth. 
\end{lemma}
\begin{proof}
Let $f$ be the number of distinct elements on an input stream. 
Given a stream $A$, it is easy to see that $f$ preserves the following properties of smoothness: (1) $f(A)\geq 0$, (2) $f(A) \geq f(B)$ where $B \subseteq A$, (3) $f(A) \leq \poly(n)$. 

It remains to show that if $B \subseteq A$ and $(1-\rho)f(A) \leq f(B)$ then $(1-\rho) f(A \cup C) \leq f(B \cup C)$ for any adjacent $C$. 
Note that $(1-\rho) f(A \cup C) \leq (1-\rho)( f(A) + f(C)) \leq f(B) + f(C) \leq f(B \cup C)$, where the first step is true because $A$ and $C$ are adjacent streams and the number of distinct elements in their union is therefore at most the sum of distinct elements in both. 
\end{proof}

We apply our general framework to the following insertion-only streaming algorithm given by \cite{Blasiok20}: 
\begin{theorem}~\cite{Blasiok20}
Given an accuracy parameter $\alpha>0$ and a failure probability $\delta \in (0,1)$, there exists a streaming algorithm that outputs a $(\alpha,0)$-approximation to the number of distinct elements with probability $1-\delta$ using $\O{\frac{1}{\alpha^2}\log\frac{1}{\delta} + \log n}$ bits.
\end{theorem}

\begin{corollary} \corlab{cor:distinct}
For any integer $K>0$ and for any $\alpha >0$, $c>0$, $\eps>0$, and $\rho:= \O{\frac{\alpha\eta}{\log m}}$ on a stream $\frakS$ of length $m=\poly(n)$, where $\eta=\min(\eps,1)$, there exists a sliding window algorithm $\calA$ such that for the dataset $W$ induced by the stream:
\begin{enumerate}
\item (Privacy) The algorithm $\calA$ is $(\eps,\delta)$-differentially private where $\eps$ is constant and $\delta=\frac{1}{m^c}$.  
\item (Accuracy) The algorithm $\calA$ outputs $\widehat{\dist}(\calW)$ such that with probability at least $1-\frac{1}{m^c}-\frac{1}{m^{c-1}(e^\eps-1)/K+1}$, 
\[(1-\alpha)\dist(\calW)-\frac{12c\log m}{\eps}-\frac{1}{Km}\le\widehat{\dist}(\calW)\le(1+\alpha)\dist(\calW)+\frac{12c\log m}{\eps} + \frac{1}{Km},\]
where $\dist(\calW)$ denotes the number of distinct elements in the last $\min(m,W)$ updates of the stream $\frakS$, for a window parameter $W>0$. 
\item (Time/Space) \sloppy The algorithm $\calA$ uses $\O{\frac{1}{\alpha^3\eta^3}\log^5 n}$ bits.
\end{enumerate}
\end{corollary}
We remark that the dependency on $\alpha$ and $\eta$ can be improved to $\frac{1}{\alpha\eta}$ with a worse dependency on $\log n$ by using an algorithm of \cite{BravermanGLWZ18} that achieves a $(\alpha,0)$-approximation to the number of distinct elements in the sliding window model, without using the smooth histogram framework.

\paragraph{$L_p$ norm and $F_p$ moment estimation.} 
We remark that $\cite{BravermanO10}$ showed that the $L_p$ norm and $F_p$ moment estimation problems are $(\alpha, \alpha^p/p)$-smooth for $p \geq 1$ and $(\alpha,\alpha)$-smooth for $0< p \leq 1$. 
Therefore, we can apply the smooth histogram framework to the $L_p$ norm and $F_p$ moment estimation problem and then further apply our general framework to the smooth histogram framework. 
However, it turns out this approach gives a sub-optimal result, since there exist the following more efficient algorithms for $L_p$ norm estimation in the sliding window model:

\begin{lemma}
\cite{WoodruffZ21}
Given $\alpha$ and $p\in(0,2]$, there exists a one-pass algorithm in the sliding window model that outputs a $(\alpha,0)$-approximation to the $L_p$ norm/$F_p$ moment with probability at least $1-\frac{1}{\poly(n)}$. 
The algorithm uses $\O{\frac{1}{\alpha^2}\log^3 n\log^3\frac{1}{\alpha}}$ space for $p=2$ and $\O{\frac{1}{\alpha^2}\log^3 n(\log\log n)^2\log^3\frac{1}{\alpha}}$ space for $p\in(0,2)$.
\end{lemma}
\begin{corollary} \corlab{cor:lp:sw}
For any integer $K>0$, $p\in(0,2]$, and for any $\alpha>0$, $c>0$, $\eps>0$, and $\rho:= \O{\frac{\alpha\eta}{\log m}}$ on a stream $\frakS$ of length $m=\poly(n)$, where $\eta=\min(\eps,1)$. 
Then there exists a sliding window algorithm $\calA$ such that:
\begin{enumerate}
\item (Privacy) The algorithm $\calA$ is $\eps$-differentially private where $\eps$ is constant.
\item (Accuracy) The algorithm $\calA$ outputs $\widehat{X}$ such that with probability at least $1-\frac{1}{m^c}-\frac{1}{\poly(m)(e^\eps-1)/K+1}$, 
\[(1-\alpha)X-\frac{12c\log m}{\eps}-\frac{1}{K\poly(m)}\le\widehat{X}\le(1+\alpha)X+\frac{12c\log m}{\eps}+\frac{1}{K\poly(m)},\]
where $X$ is the $L_p$ norm of the frequency vector induced by the most recent $\min(m,W)$ updates of the stream for $p\in(1,2]$ and $X$ is the $F_p$ frequency moment of the same frequency vector for $p\in(0,1)$.
\item (Time/Space) The algorithm uses $\tO{\frac{1}{\alpha^2\eta^2}\log^5 n\log^3\frac{1}{\alpha\eta}}$ space for $p=2$ and $\tO{\frac{1}{\alpha^2\eta^2}\log^5 n}$ space for $p\in(0,2)$.
\end{enumerate}
\end{corollary}

\section{Conclusion and Open Questions}
In this work, we introduce a general framework for transforming a non-private approximation algorithm into a differentially private approximation algorithm. 
We show specific applications of our framework for sublinear time and sublinear space algorithms. 
Although our framework applies to a large variety of problems and settings, it does incur a small penalty in both runtime and space for achieving differential privacy. 
A natural question is whether these losses are necessary for a general black-box framework and what are sufficient conditions for achieving a black-box reduction. 

It also seems possible that our framework could provide a method for achieving differentially private algorithms when the important resource is not runtime, number of queries, or space. 
For example, in distributed algorithms, it is often desired to achieve sublinear communication while in learning/testing, it is often desired to achieve sublinear query complexity. 
We believe that exploring the limits and capabilities of our framework in those settings would be a natural future direction of work.


\bibliographystyle{alpha}
\bibliography{references}

\appendix

\section{Proof of Approximate DP to Pure DP transformation}\label{app:epsdp-eps}
\dptransform*
\begin{proof}
Note that WLOG we can assume that $\calA_f(D)$ outputs a value between 0 and $M$ since we can always truncate the output to this range --- this operation preserves privacy by postprocessing and does not adversely affect accuracy. For some $K>0$, define algorithm $\calA''_f(D)$ as outputting $\frac{\lceil K \calA_f(D) \rceil }{KM}$. Observe that $\calA''_f$ is $(\eps,\delta)$-DP by postprocessing and the accuracy guarantee of $\calA''_f$ is almost identical to that of $\calA_f$ since by definition $\vert \calA''_f(D) - \calA_f(D) \vert < \frac{1}{KM}$. By post-processing we can ensure that the range $\calR$ of $\calA''_f(D)$ is small $\vert \calR\vert = (M+1)K$ since $\calR=\{\frac{i}{KM} : 0 \leq i\leq KM\}$. Thus, we can pick $p$ such that $\delta \leq \frac{(e^\eps -1)p}{\vert \calR \vert (1-p)}$ and apply a folklore theorem (see~\thmref{thm:epsdel-eps}) to transform our $(\eps,\delta)$-DP algorithm $\calA''_f(D)$ to an $\eps$-DP algorithm $\calA'_f(D)$ in the following manner:
\begin{align*}
    \calA'_f(D) = \begin{cases} 
    \calA''_f(D) & \text{with probability }1-p \\
    \mathsf{random}(\calR) &\text{with probability }p
    \end{cases}
\end{align*}
By combining the accuracy guarantees of $\calA_f$ and $\calA''_f$ we see that with probability $1-\eta - p$, we have that $(1-\alpha) f(D) - \kappa - \frac{1}{KM} \leq \calA_f(D) \leq (1+\alpha) f(D) + \kappa + \frac{1}{KM}$ where $p = \frac{\delta K(M+1)}{e^\eps - 1 +\delta K(M+1)}$ as claimed. 
\end{proof}

\thmapproxtopure*
 \begin{proof}

Recall that we define $\calA'$ as follows: 
\begin{align*}
    \calA'(D) = \begin{cases} 
    \calA(D) & \text{with probability }1-p \\
    \mathsf{random}(\calR) &\text{with probability }p
    \end{cases}
\end{align*}

   Let $D,D' \in \calD$ be neighboring databases and fix output $y \in \calR$. We first give a general claim regarding the probability of $\calA'(D)=y$ in terms of the $\Pr[\calA(D)=y]$. 
    \begin{claim} 
    For $D \in \calD$, 
     \begin{align*}
         \Pr[\calA'(D) = y] = \Pr[\calA(D)=y] \left( 1 - p \right) + \frac{p}{\vert \calR \vert}
     \end{align*}   
    \end{claim}

    Now we need to show that $\Pr[\calA'(D) = y] \leq e^\eps \Pr[\calA'(D')=y]$.
    \begin{align}
        &\Pr[\calA'(D) = y] \nonumber\\
        &= \Pr[\calA(D)=y] \left( 1 - p \right) + \frac{p}{\vert \calR \vert} \nonumber\\
        &\leq \left(1-p \right) (e^\eps \Pr[\calA(D')=y] + \delta) + \frac{p}{\vert \calR \vert} \nonumber\\
        &\leq e^\eps \Pr[\calA(D')=y]\left(1-p \right)  + \delta\left(1-p \right)+ \frac{p}{\vert R \vert} \label{eqtrans1}\\
      &= e^\eps (\Pr[\calA'(D')=y] - \frac{p}{\vert \calR \vert}) +\delta\left(1-p\right)+ \frac{p}{\vert \calR \vert} \label{eqtrans2}\\
      &\leq  e^\eps \Pr[\calA'(D')=y] +\delta\left(1 - p\right)+\frac{p}{\vert \calR \vert}(1-e^\eps) \nonumber\\
      &\leq e^\eps \Pr[\calA'(D') = y] \label{eqfinal} 
    \end{align}
The transition \ref{eqtrans1} to \ref{eqtrans2} follows from the observation that $\Pr[\calA'(D')=y] = (1-p)\Pr[\calA(D')=y] + \frac{p}{\vert \calR \vert}$ and therefore, $(1-p)\Pr[\calA(D')=y]=\Pr[\calA'(D')=y]-\frac{p}{\vert \calR \vert}$. The last equation \ref{eqfinal} follows because $\delta \leq \frac{(e^\eps-1)p}{\vert \calR \vert (1-p)}$ and thus
\[\delta\left(1 - p\right)+\frac{p}{\vert \calR \vert}(1-e^\eps) \leq 0 \ . \]
 \end{proof}
 
\end{document}